\documentclass[12pt]{article}
\usepackage{natbib}
\usepackage{multirow}

\usepackage{amsthm}
\usepackage{amssymb}
\usepackage{xcolor}
\usepackage{bbm}
\usepackage{lscape}
\usepackage{caption}
\usepackage{subcaption}
\usepackage{mathtools}
\usepackage{url}
\usepackage{float}
\usepackage{amsmath}
\usepackage{mathalfa}

\usepackage{algorithm}
\usepackage[parfill]{parskip}
\usepackage{algorithmic}
\usepackage[T1]{fontenc}
\usepackage{titling}
\usepackage{listings}
\usepackage{geometry}
\usepackage{dsfont}
\usepackage{wrapfig}
\usepackage{graphicx}
\usepackage[english]{babel}
\newtheorem{theorem}{Theorem}[section]

\numberwithin{equation}{section}

\theoremstyle{definition}

\definecolor{ao}{rgb}{0.0, 0.5, 0.0}

\title{A Multi-Stage Drop-the-Loser Design with Superiority Boundaries}
\author{Peter Greenstreet, Manel Khan, Salmaan Kanji, Pouya Motazedian, \\
Andrew Seely, Stephanie Sibley, Tim Ramsay}
\date{}
\begin{document}

\maketitle

\abstract{Multi-arm multi-stage (MAMS) trials have gained popularity, due to their improved efficiency in evaluating multiple treatments. A traditional MAMS trial often decreases the expected sample size of the trial compared to just running a multi-arm approach, but with the drawback of an increase in maximum sample size. For academic led trials this poses a particular challenge, as funding is typically based on the maximum required sample size. To address this, drop-the-loser designs were introduced, where a fixed number of treatments are dropped at each interim stage, thereby reducing the maximum sample size. In this work, we propose an enhanced multi-stage drop-the-loser design that also allows for early stopping of the entire trial for superiority. This approach aims to retain the benefits of a reduced maximum sample size while also lowering the expected sample size. The proposed design is motivated by a trial in atrial fibrillation.  We derive analytical expressions for the type I error rate, power, and expected sample size, and compare the proposed design's performance to alternative methods. \textcolor{black}{We outline the key requirements for implementing the proposed design and discuss the contexts in which it should be considered.} \textcolor{black}{For the motivating example the results show that the proposed design substantially reduces the expected sample size compared to a standard drop-the-loser design, while lowering the maximum sample size relative to running a traditional MAMS trial or multiple separate trials.}}

\section{Introduction}\label{sec1}
Multi-arm multi-stage trials have the potential to reduce the duration and large cost of clinical trials and therefore have become increasingly popular \citep{StallardNigel2020EADf, NoorNurulaminM2022Uotm, MullardAsher2018Hmdp}. Traditional multi-arm multi-stage (MAMS) trials involve comparing multiple active treatments to a common control treatment at predefined interim stages \citep{WasonJamesM.S.2012Odom, RoystonPatrick2003Ndfm, UrachS.2016Mgsd, SerraAlessandra2022Aorm, greenstreet2023change,greenstreet2023preplanned}. The interim stages can be defined to allow for the stopping of a treatment early for superiority or futility or both. Additionally in some designs the entire trial can stop early for superiority \citep{GreenstreetPeter2021Ammp, MagirrD.2012AgDt}.
\par
\textcolor{black}{While traditional MAMS designs are efficient in terms of expected sample size, they often require a larger maximum sample size compared to running a multi-arm study. Academic investigator led trials can find this design feature particularly difficult, as normal funding mechanisms are often not flexible enough to accommodate variable sample sizes \citep{kairalla2012adaptive}. As a result, applicants typically need to request funding for the maximum possible sample size, so making the trials appear disproportionately expensive \citep{wason2017multi}.} 
\par

This led to the development of the drop-the-loser design, where at each interim a fixed number of active treatments are dropped, therefore reducing the maximum sample size of the design. A well studied approach is the two-stage drop-the-loser design, in which a single interim analysis is used to only let the top performing treatment continue to the second stage \citep{sampson2005drop}. \cite{thall1989two} propose a two-stage design where only one treatment may continue to the second stage and it must also demonstrate sufficient efficacy, otherwise, the trial is terminated. Two-stage drop-the-loser designs with further flexibility have also been proposed using closed testing procedures and combination tests \citep{bretz2006confirmatory, schmidli2006confirmatory}. These two stage designs have also been extended to consider multiple doses and historical control arms with different endpoints \citep{abbas2022two, joshua2010some}. The idea of allowing for multiple interim analyses in which treatments can be dropped was considered in \cite{stallard2008group} with this work focusing on controlling the family-wise error of these designs in a conservative manner. \cite{wason2017multi}, propose a multi-stage drop-the-loser approach in which a set number of treatments is dropped at the end of each stage, with the final treatment being tested for superiority. 
\par
\textcolor{black}{Motivated by a academic investigator led trial into Post Operative Preventative Therapy for AtRial fibrillation after Thoracic Surgery (POPTARTS), we developed methodology building on \cite{wason2017multi} for a drop-the-loser design which also allows for the termination of the entire trial if all remaining treatments are found superior. This design therefore aims to gain the benefits of lower maximum sample size of the drop-the-loser design as well as decrease expected sample size seen in traditional MAMS designs. The motivating trial is further discussed in Section \ref{Sec:Mot}.}
\par
\textcolor{black}{As seen in the motivating example this type of design can be used when one has multiple active treatments of interest, a common control arm and the same primary outcome of interest for each treatment. This design supports the evaluation of multiple hypotheses with one for each active treatment. As discussed above this design could be used when funding mechanisms are based on the maximum sample size however there is a desire to be able to stop the trial early if all remaining treatments are found superior.}
\par
Building on this trial we present the methodology to calculate the required stopping boundaries and sample size for this type of multi-stage superiority drop-the-loser design in order to control the pairwise error rate (PWER) \citep{Choodari-OskooeiBabak2020Anea, SydesMatthewR2009Iiam, HowardDenaR2021Apti} and the power under the least favorable configuration (LFC) \citep{MagirrD.2012AgDt, PushpakomSudeepP2015TaIR, wason2017multi}. \textcolor{black}{The PWER is used as this ensures that the probability of making a type I error for a given active treatment is controlled at the desired level. PWER is the focus as the motivating trial is designed to evaluate distinct treatments, therefore, it is argued that PWER should be used \citep{HowardDenaR2021Apti,HowardDenaR2018Romt,MolloySleF.2022Maip, parker2020non, cook1996multiplicity}. However, there are scenarios where multiplicity adjustments, such as family-wise error rate (FWER) control or false discovery rate (FDR) control, may be necessary \citep{WasonJamesMS2014Cfmi} and it is worth noting that this work does not cover these error controls.}  The power under the LFC is used as it controls the probability that a treatment that has a clinically relevant effect is found if one exists when the other active treatments have an uninteresting treatment effect. The methodology given in Section \ref{Sec:Design} accommodates a trial with any number of active arms and can be used to drop multiple treatments at each interim analysis. This methodology is then applied to the motivating trial example of POPTARTS in Section \ref{Sec:Mot} and is compared to some alternative designs. Finally the paper will conclude with a discussion.  
\section{Design}
\label{Sec:Design}
Consider a clinical trial with $K$ experimental arms being tested against one common control arm. Each active treatment is tested at up to $J$ analyses, with there being an equal number of analyses as treatments. Therefore the maximum number of analyses, $J$, equals $K$. Let $n_{k,j}$ denote the number of patients recruited to treatment $k \in \{0,\hdots,K\}$ by the end of stage $j\in \{1,\hdots,J\}$. The null hypotheses of interest are $H_1: \delta_1 \leq 0, H_2:  \delta_2 \leq 0, \hdots, H_K:  \delta_K \leq 0,$ where $\delta_k$ is the difference in treatment effect between treatment $k$ and the control treatment (treatment 0). We denote the set of all $\delta_k$ as $\Delta$, so $\Delta= (\delta_1, \hdots, \delta_K)$. At each analysis $j$ for treatment $k$, $H_k$ is tested using the test statistic
\begin{equation*}
Z_{k,j} = \frac{\bar{\delta}_{k,j}}{\sqrt{V_{k,j}}},
\end{equation*}
where $\bar{\delta}_{k,j}$ is the difference in treatment effects of the observed patients on that given treatment $k$ and the control treatment, up to stage $j$ and $V_{k,j}$ is the variance of the observed difference in treatment effects. It is assumed that $Z_{k,j}$ follows a normal distribution $Z_{k,j} \sim (\frac{\delta_{k,j}}{\sqrt{V_{k,j}}},1)$.
\par
At each analysis $j \leq J-1$ the treatment with the smallest test statistic is dropped. Therefore treatment $k'$ is dropped at stage $j$ if $Z_{k',j}\leq Z_{k,j}$ for all $k \in \{1,\hdots K\}/\{m_1,$ $m_2,\hdots,m_{j-1}\}$, where $m_{i}$ is the treatment dropped at stage $i$. For the remaining treatments if $Z_{k,j} \geq u_{k,j}$ for all treatments left then the trial stops for superiority, where $u_{k,j}$ is the predefined boundary for treatment $k\in \{1,\hdots,K\}$ at stage $j\in \{1,\hdots,J\}$. 
One can use the approach defined in this paper even if one wants to drop multiple treatments at each interim. In this case one simply sets $n_{k,j}= n_{k,j'}$ and $u_{k,j}= u_{k,j'}$  for any interims that one wants to drop multiple treatments, similarly seen in \cite{wason2017multi}. 

\subsection{PWER}
The pairwise error rate (PWER) is the probability of recommending an ineffective treatment $k^\star$ by the end of the study regardless of other experimental arms in the trial \citep{bratton2016type}. The PWER is $1-P(\bigcap_{j=1}^{J}  B_{k^\star,j})$ with $B_{k,j}$ being the event $Z_{k,j} \leq u_{k,j}$ for all $k \in 1,\hdots, K$. Therefore to control the PWER at a given level $\alpha$ one needs to find $u_{k,1},\hdots,u_{k,J}$ so that $1-P(\bigcap_{j=1}^{J}  B_{k^\star,j}) \leq \alpha$. When one has equal sample size for each active treatment then the PWER is equal for every arm given they have the same boundaries, therefore one just needs to calculate $u_{1},\hdots,u_{J}$ where $u_j=u_{k,j}$ for all $k \in 1,\hdots K$. This can be calculated using the multivariate normal distribution function similar to \cite{MagirrD.2012AgDt,GreenstreetPeter2021Ammp, greenstreet2025multi}. In the Supporting Information Section 1 the multivariate normal distribution equations used to calculate the PWER is given for the motivating example. 
\par
If one controls the PWER this ensures that the type I error for a given treatment $k^\star$ is controlled no matter what the treatment effect of the other treatments are, as shown in Theorem \ref{theorem:PWER}. 
\begin{theorem}
The type I error for a given treatment $k^\star$ is guaranteed to be controlled at level $\alpha$ if $1-P(\bigcap_{j=1}^{J}  B_{k^\star,j})\leq \alpha$.
\label{theorem:PWER}
\end{theorem} 
\textcolor{black}{The proof for Theorem 1 is given in the Supporting Information Section 2. The underlying principles of the proof are that if $1-P(\bigcap_{j=1}^{J}  Z_{k^\star,j} \leq u_{k^\star,j})\leq \alpha$ then the PWER is controlled assuming no dropping of any treatment arms, but, includes accounting for the ability to stop the trial early for superiority. By additionally accounting for the possibility that treatment $k^\star$ may be dropped at an intermediate stage, the type I error for treatment $k^\star$ is further reduced. Furthermore, incorporating the fact that the arm can only stop early if all other treatments are found to be superior to the control further decreases the type I error for treatment $k^\star$. Overall, controlling the PWER therefore results in overly conservative control of the type I error for each treatment. However, this allows practical flexibility. For example,  ignoring the requirement that all treatments must be found superior to the control means that, if needed, other treatments in the trial can be stopped earlier (e.g., for safety concerns).} 
\par
\textcolor{black}{It is worth noting that we have defined a type I error for a given treatment as occurring when, at the point the trial ends, its test statistic exceeds the superiority boundary, regardless of whether that treatment is the best-performing option overall. This definition is used because, even if the treatment is not the best-performing, we still wish to avoid drawing the conclusion that it is superior to the control when it is not. Requiring that the treatment also be the best-performing option would further reduce the type I error. Originally, we require $Z_{k^\star,j} > u_j$ and $Z_{k,j} > u_j$ for all remaining treatments $k$. If we additionally require that treatment $k^\star$ be the best for a type I error to occur, then a type I error can occur only if $Z_{k^\star,j} > u_j$, $Z_{k,j} > u_j$ for all remaining treatments $k$, and $Z_{k^\star,j} > Z_{k,j}$ for all remaining treatments $k$, which is a subset of the original event. Consequently, any procedure that controls the pairwise error rate (PWER) under the original definition of type I error will also control it under this more restrictive definition that requires $k^\star$ to be the best-performing treatment.
}

\subsection{Power under LFC}
In this trial design the trial can stop early at each of the $j$ stages for superiority. This will happen when all the remaining treatments are found superior to the control. As recommended in the literature \citep{MagirrD.2012AgDt, wason2017multi, DunnettCharlesW1955AMCP, GreenstreetPeter2021Ammp} when this happens one wants to ensure that if there is a treatment that is superior to the other active treatments, it is found to be the best performing of the treatments and is therefore the one that is recommended. 
We assume that any given treatment $k'$, is recommended when (i) its test statistic and the test statistics for all other treatments left at that given stage cross the corresponding upper boundary, and (ii) its test statistic is the largest one of the remaining treatments. The sample size is found such that the probability of rejecting $H_{k'}$ achieves power $1-\beta$ when $\delta_{k'} = \theta'$ and $\delta_k = \theta_0$ for $k \neq k'$ where $\theta'$ is the minimum clinically interesting treatment effect and $\theta_0$ is the highest uninteresting treatment effect. This setting is known as the least favorable configuration (LFC). When calculating the power under the LFC one will likely assume equal sample size for each active treatment \citep{MagirrD.2012AgDt, wason2017multi} which is assumed for the remainder of this section, in this case the power under the LFC is the same for every treatment.
\par
The calculation of power under the LFC can be broken down into the event that the treatment of interest is recommended at each stage $j \in 1, \hdots, J$, defined by $\Phi_j$. Without loss of generality assume that treatment 1 is the treatment of interest.  The event that treatment 1 is found to be the best treatment and the trial stops for superiority at a given stage $j$, where $j<J$, can be split into 5 events. The first event is that treatment 1 is found superior to the control at the given stage, $A_{1,j}$, where $A_{k,j}: Z_{k,j}>u_{k,j}$. The second event is that the test statistic for treatment 1 is larger than the rest of the treatments being tested at that stage. This event equals 
\begin{equation*}
\bigcap_{k \in \{ 2,\hdots,K \}/ \{m_1,\hdots,m_{j-1}\}} C_{1,k,j}
\end{equation*}
where $m_1,..,m_{j-1}$ are the treatments that have already been dropped at earlier stages and $C_{k,k^\star,j}$ is the event that the test statistic for treatment $k$ is greater than the test statistic for treatment $k^\star$ at stage $j$, $C_{k,k^\star,j}: Z_{k,j}>Z_{k^\star,j}$. \textcolor{black}{It is worth noting that $\bigcap_{k \in \{ 2,\hdots,K \}/ \{m_1,\hdots,m_{j-1}\}}$ is the same as $\bigcap_{k \in \{ 2,\hdots,K \}}$ if $j=1$.}   
\par
The third event is all the test statistics of the other treatments being tested at the end of stage $j$ are greater than the test statistic of treatment $m_{j}$, and they are also found superior to the control treatment at stage $j$, 

\begin{equation*}
\bigcap_{k \in \{ 2,\hdots,K \}/ \{m_1,\hdots,m_j\}} \bigg{(}A_{k,j} \cap C_{k,m_j,j} \bigg{)}.
\end{equation*}
The fourth event is that treatment $m_i$ for $i=1,\hdots,j-1$ is dropped at stage $i$,
\begin{equation*}
\bigcap_{i \in \{ 1,\hdots,j-1 \}} \bigg{(} \bigcap_{k \in \{ 1,\hdots,K \}/ \{m_1,\hdots,m_i\}} C_{k,m_i,i} \bigg{)}.
\end{equation*}
The fifth event is that the trial did not stop for superiority at an earlier stage,
\begin{equation*}
\bigcap_{i \in \{ 1,\hdots,j-1 \}} \bigg{(} \bigcup_{k \in \{ 1,\hdots,K \}/ \{m_1,\hdots,m_i\}} B_{k,i} \bigg{)}.
\end{equation*}
Using these 5 events and every possible $m_1,\hdots m_j$ one can calculate $\Phi_j$. Therefore the event that treatment 1 is declared superior to the control at a given stage $j$, given $j<J$, is 
\begin{align*} 
\Phi_j= & \bigcup_{\substack{m_i \in \{2,\hdots,K\} / \{m_1,\hdots, m_{i-1} \} \\ i=1,\hdots, j }} \Bigg{[}  A_{1,j} \cap \bigcap_{k \in \{ 2,\hdots,K \}/ \{m_1,\hdots,m_{j-1}\}} C_{1,k,j} \cap  \\ & \bigcap_{k \in \{ 2,\hdots,K \}/ \{m_1,\hdots,m_j\}} \bigg{(}A_{k,j} \cap C_{k,m_j,j} \bigg{)}  \cap \bigcap_{i \in \{ 1,\hdots,j-1 \}} \bigg{(} \bigcap_{k \in \{ 1,\hdots,K \}/ \{m_1,\hdots,m_i\}} C_{k,m_i,i} \bigg{)} \\ & \cap \bigcap_{i \in \{ 1,\hdots,j-1 \}} \bigg{(} \bigcup_{k \in \{ 1,\hdots,K \}/ \{m_1,\hdots,m_i\}} B_{k,i} \bigg{)} \Bigg{]}. 
\end{align*}
\color{black}
It is worth noting that $$\bigcup_{\substack{m_i \in \{2,\hdots,K\} / \{m_1,\hdots, m_{i-1} \} \\ i=1,\hdots, j }}$$ represents the union taken over all possible permutations of $m_1,\hdots, m_{i-1}$.
\color{black}
For treatment 1 to be found superior at the final stage it must be the last treatment being tested. This therefore simplifies the calculation, so, the event that treatment 1 is declared superior to the control at stage $J$ is
 
\begin{align*}
\Phi_J= & \bigcup_{\substack{m_i \in \{2,\hdots,K\} / \{m_1,\hdots, m_{i-1} \} \\ i=1,\hdots, J-1 }} \Bigg{[}  A_{1,J} \cap \bigcap_{i \in \{ 1,\hdots,J-1 \}} \bigg{(} \bigcap_{k \in \{ 1,\hdots,K \}/ \{m_1,\hdots,m_i\}} C_{k,m_i,i} \bigg{)} \\ & \cap \bigcap_{i\in \{ 1,\hdots,J-1 \}} \bigg{(} \bigcup_{k \in \{ 1,\hdots,K \}/ \{m_1,\hdots,m_i\}} B_{k,i} \bigg{)} \Bigg{]}. 
\end{align*}
As $\Phi_1, \hdots, \Phi_J$ are mutually exclusive events, the power under the LFC is
\begin{equation*}
\sum^{J}_{j=1} P(\Phi_j),
\end{equation*}
with $P(\Phi_j)$ being calculated using the multivariate normal distribution. Supporting Information Section 3 gives the multivariate normal equations used for the motivating example. One then calculates the $n_j$ and $n_{0,j}$ to ensure the power under LFC is $\geq 1-\beta$, where $n_j=n_{k,j}$ for all $k \in 1,\hdots, K$. The maximum sample size can then be found to be,
$\max(N)=\sum^{J}_{j=1} n_j + n_{0,J}$, where $N$ is the sample size of the trial.

\subsection{Expected sample size}
\label{subsec:ESS}
The expected sample size can be calculated in a similar way to the power. First by calculating the probability of the event that the trials stops at each given stage.
The event that the trial stops at each stage $j$ is defined as  $\Psi_j$. Given $j<J$, $\Psi_j$ equals
\begin{align*} 
\Psi_j= & \bigcup_{\substack{k_i \in \{1,\hdots,K\} / \{k_1,\hdots, k_{i-1} \} \\ i=1,\hdots, j }} \Bigg{[} \bigcap_{k \in \{ 1,\hdots,K \}/ \{k_1,\hdots,k_j\}} \bigg{(}A_{k,j} \cap C_{k,k_j,j} \bigg{)} \cap \bigcap_{i \in \{ 1,\hdots,j-1 \}} \bigg{(} \\ &\bigcap_{k \in \{ 1,\hdots,K \}/ \{k_1,\hdots,k_i\}} C_{k,k_i,i} \bigg{)} \cap \bigcap_{i \in \{ 1,\hdots,j-1 \}} \bigg{(} \bigcap_{k \in \{ 1,\hdots,K \}/ \{k_1,\hdots,k_i\}} B_{k,i} \bigg{)} \Bigg{]}.
\end{align*}
The event that the trial stops at stage $J$ is
\begin{align*} 
\Psi_J=& \bigcup_{\substack{k_i \in \{1,\hdots,K\} / \{k_1,\hdots, k_{i-1} \} \\ i=1,\hdots, J }} \Bigg{[} \bigcap_{i \in \{ 1,\hdots,j-1 \}} \bigg{(} \bigcap_{k \in \{ 1,\hdots,K \}/ \{k_1,\hdots,k_i\}} C_{k,k_i,i} \bigg{)}
\\ &
\cap \bigcap_{i \in \{ 1,\hdots,j-1 \}} \bigg{(} \bigcap_{k \in \{ 1,\hdots,K \}/ \{k_1,\hdots,k_i\}} B_{k,i} \bigg{)} \Bigg{]}.
\end{align*}
To calculate the expected sample size one needs to include the sample size required for each event $P(\Psi_j)$. The expected sample size is therefore,
\begin{equation*}
E(N|\Delta)=
\sum^{J}_{j=1} \bigg{(} P(\Psi_j)  \bigg{(} \sum^{j-1}_{i=1} n_i + (K-j+1)n_j +n_{0,j} \bigg{)} \bigg{)}.
\end{equation*}
Once again $P(\Psi_j)$ can be calculated using the multivariate normal distribution, with the Supporting Information Section 4 giving the equations used for the motivating example. 
\color{black}
\subsection{Pre-specified design parameters and implementation}
\label{subsec:design parameters}
In order to use this design, one must pre-specify several key parameters based on medical expertise. One must specify the desired level of control of the PWER, $\alpha$, and the desired power under the LFC, $1-\beta$. Additionally, the number of treatments of interest should be prespecified.
\par
To calculate the power, one must also specify the minimal clinically relevant effect, the highest uninteresting effect, and the variability of the outcome. These three parameters should be determined through discussion with clinicians and by reviewing the literature. One can use the approaches outlined in \cite{whitehead2009one} to obtain the normal approximations for these parameters, which can then be used in the calculation of the stopping boundaries and the power.

Overall, to implement the approach discussed above with the defined parameters, one must first calculate the type I error for the given boundaries. This can be done using an iterative approach in which the PWER is calculated for the current boundaries, and if PWER is greater than $\alpha$, the boundaries are increased; if PWER is less than $\alpha$, the boundaries are decreased. This process is repeated until the boundaries control the PWER at a level between $\alpha$ and $\alpha - \omega$, where $\omega$ is a predefined tolerance. For example, $\omega$ was set to 0.00001 in the motivating example. Therefore, the bounds were selected to achieve PWER within the range $0.02499$ to $0.025$.

Once the boundaries are calculated, the required sample size per stage can then be determined. This can be done by calculating the power for a small sample size, then increasing the sample size one patient per stage until the study achieves at least $1-\beta$ power. Using this approach both the sample size and the boundaries are found. One can then calculate the expected sample size under any configuration of interest using the method set out in subsection \ref{subsec:ESS}.

\color{black}
\section{Motivating example}
\label{Sec:Mot}

This design was motivated by a trial in Post Operative Preventative Therapy for AtRial fibrillation after Thoracic Surgery (POPTARTS). Post-operative atrial fibrillation (POAF) is the most common arrhythmia encountered after thoracic surgery, with a prevalence of 10-15\% after lobectomy, 20-30\% after pneumonectomy, and 12-37\% after esophagectomy \citep{vaporciyan2004risk, seesing2019new}. In the short term the development of POAF is associated with worsened hemodynamic instability, increased risk of thromboembolic events, and prolonged intensive care unit and hospital stay \citep{ivanovic2014incidence}. In the long-term, development of POAF confers a three to four-fold increased risk of stroke, a four-fold increase in myocardial infarction, and a three-fold increase in mortality \citep{albini2021long, alturki2020major}. POAF is associated with significantly increased health care costs and resource utilization, adding over ten thousand dollars to the cost of care when patients develop the arrhythmia \citep{lapar2014postoperative}. Despite some evidence of efficacy of pharmacological prevention of POAF, as well as recommendations for use of prophylaxis in thoracic surgery guidelines, recent surveys have demonstrated poor uptake of these interventions by thoracic surgeons. In a 2014 survey by the American Association of Thoracic Surgeons 56\% of respondents did not routinely give any medication for the prevention of POAF \citep{frendl20142014}. Similarly, in a recent, yet to be published, survey of Canadian thoracic surgeons, 67\% of respondents stated they did not routinely prescribe pharmacologic prophylaxis for POAF. Therefore this motivated the desire to study the available pharmacologic treatments to see if it can be shown that any are superior to the current standard of care of no treatment. Due to this being a publicly funded, academic investigator led trial with a limited fixed maximum budget a multi-stage drop-the-loser design was recommended. However due to the severity of the POAF the ability to stop early for superiority was also desired. 
\par
There are three active treatments of interest, carvedilol, magnesium sulphate, amiodarone and the control treatment is the current standard of care of no treatments.   
One treatment will be dropped at each equally spaced stage, so there are an equal number of patients on each treatment at each stage. \textcolor{black}{A clinically relevant difference is categorized as a 5\% absolute decrease in atrial fibrillation - so a risk difference of interest of 5\% - and an uninteresting treatment effect being less than 1\% absolute decrease. The incidence of POAF after major thoracic surgery is approximately 12\%. Therefore a treatment with a clinically relevant effect would have incidence of POAF of less than 7\%.} The power of the trial was set to be 90\% and PWER controlled at 2.5\% one-sided. \textcolor{black}{Therefore using the approach discussed in \cite{whitehead2009one,JakiT2013Coca} the normal approximation for the binary endpoint can be used. This gives a clinically relevant effect of $\theta'=0.594$ based on the log-odds ratio and an uninteresting treatment effect of $\theta_0=0.098$ based on the log-odds ratio and an approximation for the variance of the observed difference in treatment effect of $V_{k,j}=\sigma^2(n_{k,j}^{-1}+n_{0,j}^{-1}) $ where $\sigma^2=9.47$ \citep{whitehead2009one, JakiT2013Coca, jaki2019r}.} The O’Brien and Fleming boundaries are used as they preserve a nominal significance level at the final analysis that is close to that of a single test procedure 
\citep{OBrienPeterC.1979AMTP, chen2014flexible}.
\par
The calculations were carried out using R \citep{Rref} with the method given here having the multivariate normal probabilities being calculated using the package \texttt{mvtnorm} \citep{mvtnorm}. The code is available in the Supporting Information.

\subsection{Alternative approaches}
The approach proposed in this paper was compared to five alternative designs. The first alternative design is to use a multi-stage drop-the-loser design,  without early termination for superiority. The second is to run multiple single stage trials with each one comparing an active treatment to the control. Building on this we will also consider the alternative design of running separate trials but now with 3 stages for each trial. At each stage the treatments will be tested for superiority using the O'Brien and Fleming boundaries. We will consider two versions of this approach, the first is using symmetric futility boundaries and the second is using futility boundaries set to zero as recommended in \cite{MagirrD.2012AgDt}. \textcolor{black}{Each individual trial controls power at 90\% and controls the PWER at 2.5\%. Because the trials are independent, it is possible for multiple trials to conclude that different active treatments are superior to the control. If, instead, one were to require control of the LFC across all trials simultaneously, the sample sizes shown in Table \ref{tab:Comparisondesignsall} would increase. In that scenario, the design would need to demonstrate not only that the treatment with the clinically relevant effect is superior to the control, but also that it outperforms all other treatments. However, such cross-trial calibration is not realistic, as power is typically controlled within each individual study rather than across multiple independent studies. For this reason, Section \ref{subsec:results} considers designs in which each trial is powered independently at the 90\% level.} 
\par
The next alternative design will be a multi-arm design. The final alternative design will be a multi-arm multi-stage (MAMS) design where the trial can stop early for superiority using the O'Brien and Fleming boundaries. This will follow the design given in \cite{MagirrD.2012AgDt} now with PWER control, so the trial will stop once a treatment is found superior to the control and the best performing treatment will be recommended. We will consider two versions of this MAMS design, the first using symmetric futility boundaries and the second using futility boundaries set to zero as recommended in \cite{MagirrD.2012AgDt}. \textcolor{black}{For the MAMS designs we will control the power under the LFC at 90\% and control the PWER at 2.5\%.} 

\subsection{Results}
\label{subsec:results}
The O'Brien and Fleming stopping boundaries  required to control the PWER at the 2.5\% significance level are: 
\begin{equation}
    u_1=3.47, \;\;\; u_2=2.45, \;\;\; u_3=2.00,
    \label{EQ:Bounds}
\end{equation}
\textcolor{black}{ Table \ref{tab:Comparisondesignsall} reports the power under the LFC, the type I error of treatment $k^\star$ under the global null, and the PWER. The PWER is equal to the type I error of treatment $k^\star$ when its true effect is equal to that of the control and all the other active treatments are dropped from the trial before the first interim. For example, the other active treatments may be dropped early if they exhibit a strongly negative effect and therefore raise safety concerns. The equations used to calculate the type I error of treatment $k^\star$ under the global null are given in the Supporting Information Section 5. Additionally, Supporting Information Section 6 gives simulation results of the type I error for a given treatment under multiple different scenarios.}
\par
Table \ref{tab:Comparisondesignsall} shows the operating characteristics of the proposed design.  The maximum sample size is given, $\max(N)$.  The expected sample size is studied under 4 configurations. The first is all the active treatment effects are the same as the control treatment which is known as the global null hypothesis, so $\delta_k=0$ for all $k \in \{1,2,3 \}$, which we denote as $\Delta_0=(0,0,0)$. The second is $\delta_1=\theta'$ and $\delta_2=\delta_3=\theta_0$ therefore under the configuration studied when looking at the power under the LFC, which we denote as $\Delta_1=(\theta',\theta_0,\theta_0)$. The final configuration is every active treatment has a clinically relevant effect, so $\delta_k=\theta'$ for all the $k \in \{1,2,3 \}$, which we denote as $\Delta_2=(\theta',\theta',\theta')$. 
\par
As can be seen in Table \ref{tab:Comparisondesignsall} the maximum sample size of the proposed approach is 1854 patients which corresponds to 206 patients being recruited per remaining arm per stage. The expected sample size  under the null configuration is 1846.5 which is very close to the maximum sample size, as the trial will only stop if all the remaining treatments are found superior. Therefore, under the null most of the time the trial will continue to the final stage. Under the LFC however the sample size now drops by over 250 patients. This is because now there is a increased probability, 62.5\%, that the trial stops early. This reduction in expected sample size is even greater when all the treatments have a clinically relevant effect, as now there is a 83.7\% chance that the trial will stop early.

\par

The alternative design of a multi-stage drop-the-loser design with out the ability to stop early for superiority reduces the maximum sample size by 27 patients compared to the proposed approach, as shown in Table \ref{tab:Comparisondesignsall}. There is also a decrease in expected sample size of 19.5 patients when under the global null, however this design comes with a significant increase in expected sample size under the LFC, with an increase of 231.0 patients, or when all the treatments have a clinically relevant effect, with an increase of 342.3 patients.
\par

Table \ref{tab:Comparisondesignsall} also shows the operating characteristics of the other alternative designs. For the multi-arm design there is an increase in both maximum sample size and in expected sample size compared to the proposed approach. The maximum sample size increases by 422 patients. Under the LFC the expected sample size increases by 680.0. The MAMS approach with symmetric boundaries also results in a further increase in maximum sample size of 498 patients. Additionally, the expected sample sizes studied are greater than that of the proposed approach. The MAMS approach with futility boundaries of zero, has the smallest expected sample size under the null, with a decease of 237.9 compared to the proposed design. The expected sample size under both the LFC and when all the treatments have a clinically relevant effect is greater than the proposed approach, with an increase of 7 and 107.5 patients, respectively. The MAMS approach with futility boundaries of zero has an increased maximum sample size of 582 patients compared to the proposed approach.
\par
For all the separate trials designs the sample size is increased further. This is because now there is the need to recruit a different control group for each comparison.  The maximum sample size increases by 1530 patients for the single stage separate trials design compared to the proposed approach. Under the LFC the expected sample size increases by 1788.0 patients. 
Compared to the proposed approach, the maximum and expected sample sizes under the LFC increased by 1,584 and 1,577.5 patients for the multi-stage separate trials with symmetric futility boundaries, and by 1,728 and 1,002.3 patients for the design with zero futility boundaries, respectively.


\begin{table*}[!t]
\color{black}
\centering
 \caption{Operating characteristics of the proposed multi-stage drop-the-loser design along with the operating characteristics of the alternative approaches.}

\resizebox{\textwidth}{!}{\begin{tabular}{c|c|c|c|c|c|c|c}
Design & Power & Type 1 error* & PWER & $\max(N)$  & $E(N|\Delta_0)$ & $E(N|\Delta_1)$  & $E(N|\Delta_2)$   \\
\hline
Multi-stage superiority drop-the-loser & 0.901 & 0.019 & 0.025 &  1854 & 1846.5 & 1596.0 & 1484.7 \\
Multi-stage drop-the-loser & 0.901 & 0.018 & 0.025 &  1827 & 1827 & 1827 & 1827 \\
Multi-arm  & 0.900 & 0.025 & 0.025 &  2276 & 2276 & 2276 & 2276 \\
MAMS symmetric futility & 0.901 & 0.022 & 0.025 &  2352 & 2332.0 & 1832.7 & 1591.4 \\
MAMS zero futility & 0.900 & 0.022 & 0.025 &  2436 & 1608.6 & 1603.0 & 1592.2 \\
Separate trials & 0.900 & 0.025 & 0.025 &  3384 & 3384 & 3384 & 3384 \\
Multi-stage separate trials symmetric futility &  0.900 & 0.025 & 0.025 & 3438 & 3421.0 & 3173.5 & 2701.9 \\
Multi-stage separate trials zero futility & 0.901 & 0.025 & 0.025 & 3582 & 2229.3 & 2598.3 & 2703.0 \\
\end{tabular}}
\par
*Type 1 error under the global null for a given treatment.
\label{tab:Comparisondesignsall} 
\color{black}
\end{table*}

\section{Discussion}
Overall this paper has presented a multi-stage drop-the-loser design which allows for early stopping for superiority which was motivated from a trial in atrial fibrillation. In Section \ref{Sec:Design} generalized methodology was presented to control both the PWER and power under the LFC of the study, along with equations needed to analytically calculate the expected sample size. Section \ref{Sec:Mot} then compared the proposed approach to the alternative designs considered in the development of the trial. This section demonstrates that the proposed approach requires fewer patients than comparable multi-arm or separate trial designs. For the motivating example the early stopping for superiority does result in an increase in the maximum sample size compared to doing a drop-the-loser design with no early stopping. This increase is 27 patients, compared to a potential saving of 342.3 patients if all the treatments have a clinically relevant effect when using the superiority boundaries compared to not. There was also a potential increase in expected sample size of the proposed design compared to the MAMS design with futility boundaries equal to zero. This increase was 237.9 under the null configuration, however the maximum sample size decrease by 582 patients by using the proposed approach.
\par
\textcolor{black}{Throughout this work, the focus has been on controlling the PWER, as this trial is designed to evaluate distinct treatments, therefore, this type I error control is appropriate \citep{HowardDenaR2021Apti,HowardDenaR2018Romt,MolloySleF.2022Maip, parker2020non, cook1996multiplicity}. However, there are scenarios where multiplicity adjustments, such as family-wise error rate (FWER) control or false discovery rate (FDR) control, may be necessary \citep{WasonJamesMS2014Cfmi}. This could include cases where multiple doses of the same treatment are tested or where regulatory guidelines demand such control. The proposed approach could be applied in this setting by using a Bonferroni correction \citep{bonferroni1936teoria,dunn1961multiple}, although this would be overly conservative. Consequently, future work could focus on extending the current methodology to control these errors exactly. Additional research could also evaluate whether the drop-the-loser design can be integrated with an all-pairwise design, thereby removing the need for a control treatment \citep{WhiteheadJohn2020Eote,greenstreet2025multi}.}







\par
This paper introduces a framework for a drop-the-loser design with superiority boundaries, centered around normally distributed test statistics. Using the methodology proposed by \cite{whitehead2009one, JakiT2013Coca} this approach can accommodate other endpoints, including binary endpoints which was used for our motivating example. When applying this methodology, it is important to be aware of potential computational difficulties in the evaluation of high-dimensional multivariate normal distributions when considering trials with many arms. In such cases, one may consider the techniques proposed in \cite{BlondellLucy2021Game} for managing high-dimensional multivariate normal computations, or alternatively, use a simulation-based approach. A simulation-based approach also has a further benefit when using non-normally distributed endpoints as it can help remove the bias caused from the normality assumptions and the normality approximations used in calculating the variance and covariance matrix \citep{whitehead2009one, JakiT2013Coca, jaki2019r}. \textcolor{black}{As done in \cite{wason2017multi} in this work when calculating the covariance matrix it has been assumed that $V_{k,j}=\sigma^2(n_{k,j}^{-1}+n_{0,j}^{-1})$ where $n_{k,j}=n_{k^\star,j}$ for all $k,k^\star$. Therefore an area for further work is calculating the multivariate normal equations for when $\sigma_k \neq \sigma_{k^\star}$  where $\sigma_k^2$ is the variance of treatment $k$ and $n_{k,j} \neq n_{k^\star,j}$.} \textcolor{black}{It is worth noting that the boundaries for PWER control can be calculated using pre-existing R packages such as \textit{gsDesign} \citep{gs} and \textit{rpact} \citep{rpact}}.
\par
\textcolor{black}{To use such a design, one must have multiple active treatments of interest, a common control arm, and the same primary outcome for each treatment. To implement the design, several key parameters must be prespecified, as detailed in Subsection \ref{subsec:design parameters}. It is also worth noting, as with other drop-the-loser designs, that if all the treatments have a clinically relevant effect, there is a high probability that one of these effective treatments may be dropped from the study before reaching the superiority boundaries. Therefore, if the aim is to identify all clinically relevant treatments, one should consider whether a MAMS design would be more appropriate.}
\par
\textcolor{black}{Overall, the proposed design incorporates some of the advantages of both MAMS and drop-the-loser approaches. It retains the MAMS benefit of allowing early stopping for superiority, which can lead to substantial reductions in expected sample size. At the same time, the drop-the-loser component helps limit the maximum sample size required by systematically eliminating less promising arms during the trial.}

\section*{Author contributions}
Peter Greenstreet derived the equations and wrote the paper with the support of Manel Khan. Salmaan Kanji, Motacedian Pouya, Andrew Seely and Stephanie Sibley were key to the design of this trial with their clinical expertise, and reviewed the manuscript. Tim Ramsay oversaw the entire project and contributed greatly to the manuscript development.

\section*{Acknowledgments}
P.G. wishes to acknowledge CAN TAP TALENT for its role in supporting the completion of this manuscript. The CAN TAP TALENT CTTP is funded by the Canadian Institutes of Health Research (CIHR) – Grant \#184898.  P.G. was also supported by a CANSTAT trainee award funded by CIHR grant \#262556.

\bibliographystyle{apalike}
\bibliography{refSI} 
\section{Supporting Information}
\subsection{Equations for the proposed approach to calculate the PWER for the motivating example}
For the motivating example the PWER equals

\begin{align*}
1-\int^{u_1}_{-\infty} \int^{u_2}_{-\infty} \int^{u_3}_{-\infty}  \phi\bigg{(}\mathbf{z} ,\mu',\Sigma' \bigg{)} \textbf{d}\mathbf{z},
\end{align*}
where $\phi(\mathbf{z},\mu,\Sigma)$ is the probability density function of a multi-variate normal distribution with mean $\mu$ and covariance matrix $\Sigma$, with 
\begin{align*}
\mu'= (0,0,0)
\end{align*}
and
\begin{equation*}
\Sigma'=\begin{pmatrix}
1  & \sqrt{\frac{1}{2}} & \sqrt{\frac{1}{3}}
\\
\sqrt{\frac{1}{2}}  & 1 & \sqrt{\frac{2}{3}} 
\\
\sqrt{\frac{1}{3}}  & \sqrt{\frac{2}{3}} & 1 
\end{pmatrix}.
\end{equation*}

\color{black}
\subsection{Proof of Theorem 1}
\begin{proof}
 Let $D_{k^\star,j}$ define the event that treatment $k^\star$ is not found to be better than the control at the given stage $j$,
\begin{equation*}
D_{k^\star,j}=\left\{\begin{matrix}
Z_{k^\star,j} \leq u_{k,j} & \text{if } k^\star \text{ is still in the trial}  \\
\Omega  & \text{if } k^\star \text{ is not in the trial}\\
\end{matrix}\right. . 
\end{equation*}
where $\Omega$ is the whole sample space.
Let $m_1,\hdots m_K$ be the ordering in which the treatments are dropped from the trial. With $m_K$ being the last treatment dropped. We define $m_1 \in \{ 1,\hdots, K\}$ and define $m_k \in \{ 1,\hdots, K\}/\{m_1, \hdots, m_{k-1}\}$.
For any given $m_1,\hdots m_K$ the event that treatment $k^\star$ is not dropped from the trial equals: 
\begin{equation}
R=\bigcap_{j=1}^{J-1} \bigg{(}D_{k^\star,j} \cup \bigcup_{k=\{j+1,\hdots,J \}} B_{{m_k},j} \bigg{)} \cap \bigg{(}D_{k^\star,J} \cup B_{{m_K},J} \bigg{)}.
\label{equ:proof}
\end{equation}

R in Equation \ref{equ:proof} can be broken down into 2 key parts. The first part is  $D_{k^\star,j} \cup \bigcup_{k=\{j+1,\hdots,J \}} B_{{m_k},j} $. This states that at each stage (up-to stage J-1) the trial does not stop for superiority. For stage J this is given by the equation $D_{k^\star,J} \cup B_{{m_K},J} $.
The second part is the intercept of all these events is calculated. This is because for treatment $k^\star$ not to be found superior it must not be stopped for superiority at any stage of the trial.
It is worth noting that if $B_{{m_K},J}=B_{{k^\star},J}$ then this corresponds to the scenario in which treatment $k^\star$ would be the last treatment dropped from the trial. Furthermore $k=\{j+1,\hdots,J \}$ in $\bigcup_{k=\{j+1,\hdots,J \}} B_{{m_k},j}$ as the decision on if the trial stops for superiority at stages j<J is based on the treatment effects of the treatments left, after one of the treatments has been dropped at that stage. However the proof would still hold if one instead required all the treatments at a given stage to be superior to the control including the one that is dropped as $\bigcup_{k=\{j+1,\hdots,J \}} B_{{m_k},j} \subseteq \bigcup_{k=\{j,\hdots,J \}} B_{{m_k},j}$.
\par
Therefore the type I error of treatment $k^\star$ for given $m_1,\hdots m_K$ equals $1-P(R)$ 
As $ B_{k^\star,j} \subseteq D_{k^\star,j}$, 
\begin{equation*}
R \supseteq \bigcap_{j=1}^{J-1} \bigg{(}B_{k^\star,j} \cup \bigcup_{k=\{j+1,\hdots,J \}} B_{{m_k},j} \bigg{)} \cap \bigg{(}B_{k^\star,J} \cup B_{{m_K},J} \bigg{)} \supseteq \bigcap_{j=1}^{J-1} B_{k^\star,j} \cap B_{k^\star,J},
\end{equation*}
so
\begin{equation*}
P(R) \geq P(\bigcap_{j=1}^{J-1} B_{k^\star,j} \cap B_{k^\star,J}), 
\end{equation*}  
therefore if  $1-P(\bigcap_{j=1}^{J} B_{k^\star,j}) \leq \alpha$ then $1-P(R) \leq \alpha$. 
\end{proof}
\color{black}
\subsection{Equations for the proposed approach to calculate the Power under the LFC for the motivating example}

To calculate the power under the LFC we calculate $P(\Phi_1)$, $P(\Phi_2)$ and $P(\Phi_3)$.  As the trial has an equal number of patients per stage per arm $n$ is used as $n=n_{1,k}=n_{2,k}-n_{1,k}=n_{3,k}-n_{2,k}$ for all $k \in \{0,1,2,3\}$. $P(\Phi_1)$ equals
\begin{align*}
P(\Phi_1)& = 2 \bigg{[} \int^{\infty}_{u_1} \int^{\infty}_{0} \int^{\infty}_{0} \int^{\infty}_{u_1} \int^{\infty}_{0} \phi\bigg{(}\mathbf{z} ,\mu^{\Phi_1},\Sigma^{\Phi_1}\bigg{)} \textbf{d}\mathbf{z} \bigg{]},
\end{align*}

where
\begin{align*}
\mu^{\Phi_1}=\bigg{(} \frac{\theta'\sqrt{n}} {\sigma\sqrt{2}}, 
\frac{(\theta'-\theta_0) \sqrt{n}} {\sigma\sqrt{2}},
\frac{(\theta'-\theta_0) \sqrt{n}} {\sigma\sqrt{2}},
\frac{\theta_0 \sqrt{n}} {\sigma\sqrt{2}},0 \bigg{)},
\end{align*}
and
\begin{equation*}
\Sigma^{\Phi_1}=\begin{pmatrix}
1  & \frac{1}{2} & \frac{1}{2} & \frac{1}{2}  & 0
\\
\frac{1}{2}  & 1 & \frac{1}{2} & -\frac{1}{2}  & -\frac{1}{2}
\\
\frac{1}{2}  & \frac{1}{2} & 1 & 0 & \frac{1}{2}
\\
\frac{1}{2}  & -\frac{1}{2} & 0 & 1 & \frac{1}{2} 
\\
0  & -\frac{1}{2} & \frac{1}{2} & \frac{1}{2} & 1 
\end{pmatrix}.
\end{equation*}
\textcolor{black}{It is worth noting that for the events $C_{k,k^\star,j}$ one can use the fact this event equals $Z_{k,j}-Z_{k^\star,j}>0$, to create a well-defined space to integrate over. For $P(\Phi_1$) we are calculating the integrals across: $Z_{1,1}, Z_{1,1}-Z_{2,1}, Z_{1,1}-Z_{3,1}, Z_{2,1},Z_{2,1}-Z_{3,1}$. Therefore to calculate $\Sigma^{\Phi_1}$ one needs to find the covariance between: $Z_{1,1} \text{ and } Z_{1,1};\text{ } Z_{1,1} \text{ and } Z_{1,1}-Z_{2,1};\text{ } Z_{1,1} \text{ and } Z_{1,1}-Z_{3,1};\text{ }\hdots;\text{ } Z_{2,1} \text{ and } Z_{2,1}-Z_{3,1};\text{ } Z_{2,1}-Z_{3,1} \text{ and } Z_{2,1}-Z_{3,1}$. Subsection \ref{subsec:cov} of the Supporting Information provides generalized equations to calculate the elements of the covariance matrix when $V_{k,j}$ can be written in the form $V_{k,j} = \sigma^2 (n_{k,j} + n_{0,j})$, where $n_{k,j} = n_{k^\star,j}$ for all $k, k^\star$, as done for the motivating example. Furthermore, this formulation takes advantage of the fact that under the LFC both active treatments 2 and 3 have the same effect of interest. Thus, rather than calculating the above expression with treatment 2 being dropped first instead of treatment 3, one can simply double the integral.} 

$P(\Phi_2)$ equals
\begin{align*}
P(\Phi_2)& = 2 \bigg{[} \int^{\infty}_{u_2} \int^{\infty}_{0} \int^{\infty}_{0} \int^{\infty}_{0} \int^{u_1}_{-\infty} \int^{u_1}_{-\infty} \phi\bigg{(}\mathbf{z} ,\mu^{\Phi_2},\Sigma^{\Phi_2}\bigg{)} \textbf{d}\mathbf{z} 
\\
&  + \int^{\infty}_{u_2} \int^{\infty}_{0} \int^{\infty}_{0} \int^{\infty}_{0} \int^{\infty}_{u_1} \int^{u_1}_{-\infty} \phi\bigg{(}\mathbf{z} ,\mu^{\Phi_2},\Sigma^{\Phi_2}\bigg{)} \textbf{d}\mathbf{z} 
\\
&  + \int^{\infty}_{u_2} \int^{\infty}_{0} \int^{\infty}_{0} \int^{\infty}_{0} \int^{u_1}_{-\infty} \int^{\infty}_{u_1} \phi\bigg{(}\mathbf{z} ,\mu^{\Phi_2},\Sigma^{\Phi_2}\bigg{)} \textbf{d}\mathbf{z} \bigg{]}, 
\end{align*}

where
\begin{align*}
\mu^{\Phi_2}=\bigg{(} \frac{\theta'\sqrt{n}} {\sigma}, 
\frac{(\theta'-\theta_0) \sqrt{n}} {\sigma},
\frac{(\theta'-\theta_0) \sqrt{n}} {\sigma\sqrt{2}},0,
\frac{\theta' \sqrt{n}} {\sigma\sqrt{2}},\frac{\theta_0 \sqrt{n}} {\sigma\sqrt{2}} \bigg{)},
\end{align*}
and
\begin{equation*}
\Sigma^{\Phi_2}=\begin{pmatrix}
1  & \frac{1}{2} & \frac{1}{2\sqrt{2}} & 0 &  \frac{1}{\sqrt{2}}  & \frac{1}{2\sqrt{2}}
\\
\frac{1}{2}  & 1 & \frac{1}{2\sqrt{2}} & -\frac{1}{2\sqrt{2}} & \frac{1}{2\sqrt{2}}  & -\frac{1}{2\sqrt{2}}
\\
\frac{1}{2\sqrt{2}}  & \frac{1}{2\sqrt{2}} & 1 & \frac{1}{2} & \frac{1}{2} & 0
\\
0 & -\frac{1}{2\sqrt{2}} & \frac{1}{2}  & 1 & 0 &\frac{1}{2} 
\\
\frac{1}{\sqrt{2}} & \frac{1}{2\sqrt{2}}  & \frac{1}{2} & 0 & 1 & \frac{1}{2} 
\\
\frac{1}{2\sqrt{2}} & -\frac{1}{2\sqrt{2}}  & 0 &\frac{1}{2} & \frac{1}{2} & 1
\end{pmatrix}.
\end{equation*}
\textcolor{black}{For $P(\Phi_2$) we calculated the integrals across: $Z_{1,2}, Z_{1,2}-Z_{2,2}, Z_{1,1}-Z_{3,1}, Z_{2,1}-Z_{3,1}, Z_{1,1},Z_{2,1}$.} 
$P(\Phi_3)$ equals
\begin{align*}
P(\Phi_3)& = 2 \bigg{[} \int^{\infty}_{u_3} \int^{\infty}_{0} \int^{\infty}_{0} \int^{\infty}_{0} \int^{u_1}_{-\infty} \int^{u_1}_{-\infty} \int^{u_2}_{-\infty}  \phi\bigg{(}\mathbf{z} ,\mu^{\Phi_3},\Sigma^{\Phi_3}\bigg{)} \textbf{d}\mathbf{z}  
\\
&  + \int^{\infty}_{u_3} \int^{\infty}_{0} \int^{\infty}_{0} \int^{\infty}_{0} \int^{\infty}_{u_1} \int^{u_1}_{-\infty} \int^{u_2}_{-\infty}  \phi\bigg{(}\mathbf{z} ,\mu^{\Phi_3},\Sigma^{\Phi_3}\bigg{)} \textbf{d}\mathbf{z}
\\
&  + \int^{\infty}_{u_3} \int^{\infty}_{0} \int^{\infty}_{0} \int^{\infty}_{0} \int^{u_1}_{-\infty} \int^{\infty}_{u_1} \int^{u_2}_{-\infty} \phi\bigg{(}\mathbf{z} ,\mu^{\Phi_3},\Sigma^{\Phi_3}\bigg{)} \textbf{d}\mathbf{z} \bigg{]}, 
\end{align*}

where
\begin{align*}
\mu^{\Phi_3}=\bigg{(} \frac{\theta'\sqrt{3n}} {\sigma\sqrt{2}}, \frac{(\theta'-\theta_0) \sqrt{n}} {\sigma\sqrt{2}},0,
\frac{(\theta'-\theta_0) \sqrt{n}} {\sigma},
\frac{\theta' \sqrt{n}} {\sigma\sqrt{2}},\frac{\theta_0 \sqrt{n}} {\sigma\sqrt{2}} , \frac{\theta' \sqrt{n}} {\sigma} \bigg{)},
\end{align*}
and
\begin{equation*}
\Sigma^{\Phi_3}=\begin{pmatrix}
1  & \frac{1}{2\sqrt{3}} & 0 & \frac{1}{\sqrt{6}} & \frac{1}{\sqrt{3}} & \frac{1}{2\sqrt{3}} &  \frac{\sqrt{2}}{\sqrt{3}}
\\
\frac{1}{2\sqrt{3}} & 1 & \frac{1}{2} & \frac{1}{2\sqrt{2}} &  \frac{1}{2} & 0  & \frac{1}{2\sqrt{2}}
\\
0 & \frac{1}{2} & 1 & -\frac{1}{2\sqrt{2}}  & 0 & \frac{1}{2} & 0
\\
\frac{1}{\sqrt{6}} & \frac{1}{2\sqrt{2}} & -\frac{1}{2\sqrt{2}} & 1 & \frac{1}{2\sqrt{2}}   & -\frac{1}{2\sqrt{2}}  &\frac{1}{2} 
\\
\frac{1}{\sqrt{3}} &\frac{1}{2} & 0 & \frac{1}{2\sqrt{2}} & 1 & \frac{1}{2} & \frac{1}{\sqrt{2}} 
\\
\frac{1}{2\sqrt{3}} &0 & \frac{1}{2}  & -\frac{1}{2\sqrt{2}} &\frac{1}{2} & 1 & \frac{1}{2\sqrt{2}}
\\
\frac{\sqrt{2}}{\sqrt{3}} & \frac{1}{2\sqrt{2}} & 0 &  \frac{1}{2} & \frac{1}{\sqrt{2}} & \frac{1}{2\sqrt{2}} & 1
\end{pmatrix}.
\end{equation*}
\textcolor{black}{For $P(\Phi_3$) we calculated the integrals across: $Z_{1,3}, Z_{1,1}-Z_{3,1}, Z_{2,1}-Z_{3,1}, Z_{1,2}-Z_{2,2}, Z_{1,1},Z_{2,1},Z_{1,2}$.} 
The power under the LFC for the motivating example is therefore 
\begin{equation*}
\sum^{3}_{j=1} P(\Phi_j).
\end{equation*}
\color{black}
\subsubsection{General equation for covariance matrix}
\label{subsec:cov}
Under the same assumptions as used in \cite{wason2017multi} of $V_{k,j}=\sigma^2(n_{k,j}^{-1}+n_{0,j}^{-1})$ where $n_{k,j}=n_{k^\star,j}$ for all $k,k^\star$ the covariance between the events $B_{k,j}$ and $B_{k^\star,j^\star}$; or $B_{k,j}$ and $A_{k^\star,j^\star}$; or $A_{k,j}$ and $B_{k^\star,j^\star}$; or  $A_{k,j}$ and $A_{k^\star,j^\star}$ equals: 

\begin{equation*}
\text{Cov}(Z_{k,j},Z_{k^\star,j^\star})=\left\{\begin{matrix}
\sqrt{\frac{\min(n_j,n_{j^\star})}{\max(n_j,n_{j^\star})}} & \text{if } k^\star= k \\
\frac{1}{2}\sqrt{\frac{\min(n_j,n_{j^\star})}{\max(n_j,n_{j^\star})}}  & \text{if } k^\star \neq k \\
\end{matrix}\right. . 
\end{equation*}

The covariance between the events $C_{k_1,k_2,j}$ and $C_{k^\star_1,k^\star_2,j^\star}$ equals: 
\begin{equation*}
\text{Cov}(Z_{k_1,j}-Z_{k_2,j},Z_{k^\star_1,j^\star}-Z_{k^\star_2,j^\star})=\left\{\begin{matrix}
\sqrt{\frac{\min(n_j,n_{j^\star})}{\max(n_j,n_{j^\star})}} & \text{if } k_1^\star= k_1 \;\& \; k_2^\star= k_2  \\
-\sqrt{\frac{\min(n_j,n_{j^\star})}{\max(n_j,n_{j^\star})}} & \text{if } k_2^\star= k_1 \;\& \; k_1^\star= k_2 \\
\frac{1}{2}\sqrt{\frac{\min(n_j,n_{j^\star})}{\max(n_j,n_{j^\star})}}  & \text{if } k_1^\star= k_1 \;\& \; k_2^\star \neq k_2 \text{ or }   k_1^\star \neq k_1 \;\& \; k_2^\star= k_2\\
-\frac{1}{2}\sqrt{\frac{\min(n_j,n_{j^\star})}{\max(n_j,n_{j^\star})}}  & \text{if } k_1^\star= k_2 \;\& \; k_2^\star \neq k_1 \text{ or }   k_1^\star \neq k_2 \;\& \; k_2^\star= k_1\\
0  & \text{if } k_1^\star \neq k_1 \;\& \; k_1^\star \neq k_2 \;\& \; k_2^\star \neq k_1 \;\& \; k_2^\star \neq k_2\\
\end{matrix}\right. . 
\end{equation*}
It is worth noting that by design $k_1 \neq k_2$ and $k^\star_1 \neq k^\star_2$. 
The covariance between the events $A_{k_1,j}$ and $C_{k^\star_1,k^\star_2,j^\star}$; or $B_{k_1,j}$ and $C_{k^\star_1,k^\star_2,j^\star}$ equals:
\begin{equation*}
\text{Cov}(Z_{k_1,j},Z_{k^\star_1,j^\star}-Z_{k^\star_2,j^\star})=\left\{\begin{matrix}
\frac{1}{2}\sqrt{\frac{\min(n_j,n_{j^\star})}{\max(n_j,n_{j^\star})}}  & \text{if } k_1 = k_1^\star  \\
-\frac{1}{2}\sqrt{\frac{\min(n_j,n_{j^\star})}{\max(n_j,n_{j^\star})}}  & \text{if }  k_1 = k_2^\star\\
0  & \text{if } k_1 \neq k_1^\star \;\& \; k_1 \neq k_2^\star\\
\end{matrix}\right. . 
\end{equation*}
\color{black}
\subsection{Equations for the proposed approach to calculate the expected sample size for the motivating example}

To calculate the expected sample size for the given $\Delta$ we calculate $P(\Psi_1)$, $P(\Psi_2)$ and $P(\Psi_3)$.  $P(\Psi_1)$ equals

\begin{align*}
P(\Psi_1)& = \int^{\infty}_{u_1} \int^{\infty}_{0}  \int^{\infty}_{u_1} \int^{\infty}_{0} \phi\bigg{(}\mathbf{z} ,\mu^{\Psi_{1, a}},\Sigma^{\Psi_1}\bigg{)} \textbf{d}\mathbf{z} + \int^{\infty}_{u_1} \int^{\infty}_{0}  \int^{\infty}_{u_1} \int^{\infty}_{0} \phi\bigg{(}\mathbf{z} ,\mu^{\Psi_{1, b}},\Sigma^{\Psi_1}\bigg{)} \textbf{d}\mathbf{z} 
\\
&+\int^{\infty}_{u_1} \int^{\infty}_{0}  \int^{\infty}_{u_1} \int^{\infty}_{0} \phi\bigg{(}\mathbf{z} ,\mu^{\Psi_{1, c}},\Sigma^{\Psi_1}\bigg{)} \textbf{d}\mathbf{z},
\end{align*}

where
\begin{align*}
\mu^{\Psi_{1,a}}=\bigg{(} \frac{\delta_{i_1} \sqrt{n}} {\sigma\sqrt{2}}, 
\frac{(\delta_{i_1}-\delta_{i_3}) \sqrt{n}} {\sigma\sqrt{2}},
\frac{\delta_{i_2} \sqrt{n}} {\sigma\sqrt{2}},
\frac{(\delta_{i_2}-\delta_{i_3}) \sqrt{n}} {\sigma\sqrt{2}} \bigg{)},
\end{align*}
with $i_1=1,i_2=2,i_3=3$;
\begin{align*}
\mu^{\Psi_{1,b}}=\bigg{(} \frac{\delta_{i_1} \sqrt{n}} {\sigma\sqrt{2}}, 
\frac{(\delta_{i_1}-\delta_{i_3}) \sqrt{n}} {\sigma\sqrt{2}},
\frac{\delta_{i_2} \sqrt{n}} {\sigma\sqrt{2}},
\frac{(\delta_{i_2}-\delta_{i_3}) \sqrt{n}} {\sigma\sqrt{2}} \bigg{)},
\end{align*}
with $i_1=1,i_2=3,i_3=2$;
\begin{align*}
\mu^{\Psi_{1,c}}=\bigg{(} \frac{\delta_{i_1} \sqrt{n}} {\sigma\sqrt{2}}, 
\frac{(\delta_{i_1}-\delta_{i_3}) \sqrt{n}} {\sigma\sqrt{2}},
\frac{\delta_{i_2} \sqrt{n}} {\sigma\sqrt{2}},
\frac{(\delta_{i_2}-\delta_{i_3}) \sqrt{n}} {\sigma\sqrt{2}} \bigg{)},
\end{align*}
with $i_1=2,i_2=3,i_3=1$ and
\begin{equation*}
\Sigma^{\Psi_1}=\begin{pmatrix}
1  & \frac{1}{2} & \frac{1}{2}  & 0
\\
\frac{1}{2}  & 1 & 0 & \frac{1}{2}
\\
\frac{1}{2}  & 0 & 1 & \frac{1}{2}
\\
0 & \frac{1}{2} & \frac{1}{2} & 1
\end{pmatrix}.
\end{equation*}

$P(\Psi_2)$ equals
\begin{align*}
P(\Psi_2)& =  \int^{\infty}_{u_2} \int^{\infty}_{0} \int^{\infty}_{0} \int^{\infty}_{0} \int^{u_1}_{-\infty} \int^{u_1}_{-\infty} \phi\bigg{(}\mathbf{z} ,\mu^{\Psi_{2,a}},\Sigma^{\Psi_2}\bigg{)} \textbf{d}\mathbf{z}  
\\
& + \int^{\infty}_{u_2} \int^{\infty}_{0} \int^{\infty}_{0} \int^{\infty}_{0} \int^{\infty}_{u_1} \int^{u_1}_{-\infty} \phi\bigg{(}\mathbf{z} ,\mu^{\Psi_{2,a}},\Sigma^{\Psi_2}\bigg{)} \textbf{d}\mathbf{z} 
\\
& + \int^{\infty}_{u_2} \int^{\infty}_{0} \int^{\infty}_{0} \int^{\infty}_{0} \int^{u_1}_{-\infty} \int^{\infty}_{u_1} \phi\bigg{(}\mathbf{z} ,\mu^{\Psi_{2,a}},\Sigma^{\Psi_2}\bigg{)} \textbf{d}\mathbf{z} 
\\
& + \int^{\infty}_{u_2} \int^{\infty}_{0} \int^{\infty}_{0} \int^{\infty}_{0} \int^{u_1}_{-\infty} \int^{u_1}_{-\infty} \phi\bigg{(}\mathbf{z} ,\mu^{\Psi_{2,b}},\Sigma^{\Psi_2}\bigg{)} \textbf{d}\mathbf{z}  
\\
& + \int^{\infty}_{u_2} \int^{\infty}_{0} \int^{\infty}_{0} \int^{\infty}_{0} \int^{\infty}_{u_1} \int^{u_1}_{-\infty} \phi\bigg{(}\mathbf{z} ,\mu^{\Psi_{2,b}},\Sigma^{\Psi_2}\bigg{)} \textbf{d}\mathbf{z} 
\\
& + \int^{\infty}_{u_2} \int^{\infty}_{0} \int^{\infty}_{0} \int^{\infty}_{0} \int^{u_1}_{-\infty} \int^{\infty}_{u_1} \phi\bigg{(}\mathbf{z} ,\mu^{\Psi_{2,b}},\Sigma^{\Psi_2}\bigg{)} \textbf{d}\mathbf{z} 
\\
& + \int^{\infty}_{u_2} \int^{\infty}_{0} \int^{\infty}_{0} \int^{\infty}_{0} \int^{u_1}_{-\infty} \int^{u_1}_{-\infty} \phi\bigg{(}\mathbf{z} ,\mu^{\Psi_{2,c}},\Sigma^{\Psi_2}\bigg{)} \textbf{d}\mathbf{z}  
\\
& + \int^{\infty}_{u_2} \int^{\infty}_{0} \int^{\infty}_{0} \int^{\infty}_{0} \int^{\infty}_{u_1} \int^{u_1}_{-\infty} \phi\bigg{(}\mathbf{z} ,\mu^{\Psi_{2,c}},\Sigma^{\Psi_2}\bigg{)} \textbf{d}\mathbf{z} 
\\
& + \int^{\infty}_{u_2} \int^{\infty}_{0} \int^{\infty}_{0} \int^{\infty}_{0} \int^{u_1}_{-\infty} \int^{\infty}_{u_1} \phi\bigg{(}\mathbf{z} ,\mu^{\Psi_{2,c}},\Sigma^{\Psi_2}\bigg{)} \textbf{d}\mathbf{z} 
\\
& + \int^{\infty}_{u_2} \int^{\infty}_{0} \int^{\infty}_{0} \int^{\infty}_{0} \int^{u_1}_{-\infty} \int^{u_1}_{-\infty} \phi\bigg{(}\mathbf{z} ,\mu^{\Psi_{2,d}},\Sigma^{\Psi_2}\bigg{)} \textbf{d}\mathbf{z}  
\\
& + \int^{\infty}_{u_2} \int^{\infty}_{0} \int^{\infty}_{0} \int^{\infty}_{0} \int^{\infty}_{u_1} \int^{u_1}_{-\infty} \phi\bigg{(}\mathbf{z} ,\mu^{\Psi_{2,d}},\Sigma^{\Psi_2}\bigg{)} \textbf{d}\mathbf{z} 
\\
& + \int^{\infty}_{u_2} \int^{\infty}_{0} \int^{\infty}_{0} \int^{\infty}_{0} \int^{u_1}_{-\infty} \int^{\infty}_{u_1} \phi\bigg{(}\mathbf{z} ,\mu^{\Psi_{2,d}},\Sigma^{\Psi_2}\bigg{)} \textbf{d}\mathbf{z} 
\\
& + \int^{\infty}_{u_2} \int^{\infty}_{0} \int^{\infty}_{0} \int^{\infty}_{0} \int^{u_1}_{-\infty} \int^{u_1}_{-\infty} \phi\bigg{(}\mathbf{z} ,\mu^{\Psi_{2,e}},\Sigma^{\Psi_2}\bigg{)} \textbf{d}\mathbf{z}  
\\
& + \int^{\infty}_{u_2} \int^{\infty}_{0} \int^{\infty}_{0} \int^{\infty}_{0} \int^{\infty}_{u_1} \int^{u_1}_{-\infty} \phi\bigg{(}\mathbf{z} ,\mu^{\Psi_{2,e}},\Sigma^{\Psi_2}\bigg{)} \textbf{d}\mathbf{z} 
\\
& + \int^{\infty}_{u_2} \int^{\infty}_{0} \int^{\infty}_{0} \int^{\infty}_{0} \int^{u_1}_{-\infty} \int^{\infty}_{u_1} \phi\bigg{(}\mathbf{z} ,\mu^{\Psi_{2,e}},\Sigma^{\Psi_2}\bigg{)} \textbf{d}\mathbf{z} 
\\
& + \int^{\infty}_{u_2} \int^{\infty}_{0} \int^{\infty}_{0} \int^{\infty}_{0} \int^{u_1}_{-\infty} \int^{u_1}_{-\infty} \phi\bigg{(}\mathbf{z} ,\mu^{\Psi_{2,f}},\Sigma^{\Psi_2}\bigg{)} \textbf{d}\mathbf{z}  
\\
& + \int^{\infty}_{u_2} \int^{\infty}_{0} \int^{\infty}_{0} \int^{\infty}_{0} \int^{\infty}_{u_1} \int^{u_1}_{-\infty} \phi\bigg{(}\mathbf{z} ,\mu^{\Psi_{2,f}},\Sigma^{\Psi_2}\bigg{)} \textbf{d}\mathbf{z} 
\\
& + \int^{\infty}_{u_2} \int^{\infty}_{0} \int^{\infty}_{0} \int^{\infty}_{0} \int^{u_1}_{-\infty} \int^{\infty}_{u_1} \phi\bigg{(}\mathbf{z} ,\mu^{\Psi_{2,f}},\Sigma^{\Psi_2}\bigg{)} \textbf{d}\mathbf{z},
\end{align*}

where
\begin{align*}
\mu^{\Psi_{2,a}}=\bigg{(} \frac{\delta_{i_1}\sqrt{n}} {\sigma}, 
\frac{(\delta_{i_1}-\delta_{i_2}) \sqrt{n}} {\sigma},
\frac{(\delta_{i_1}-\delta_{i_3}) \sqrt{n}} {\sigma\sqrt{2}},\frac{(\delta_{i_2}-\delta_{i_3}) \sqrt{n}} {\sigma\sqrt{2}},
\frac{\delta_{i_1} \sqrt{n}} {\sigma\sqrt{2}},\frac{\delta_{i_2} \sqrt{n}} {\sigma\sqrt{2}} \bigg{)},
\end{align*}
with $i_1=1,i_2=2,i_3=3$;

\begin{align*}
\mu^{\Psi_{2,b}}=\bigg{(} \frac{\delta_{i_1}\sqrt{n}} {\sigma}, 
\frac{(\delta_{i_1}-\delta_{i_2}) \sqrt{n}} {\sigma},
\frac{(\delta_{i_1}-\delta_{i_3}) \sqrt{n}} {\sigma\sqrt{2}},\frac{(\delta_{i_2}-\delta_{i_3}) \sqrt{n}} {\sigma\sqrt{2}},
\frac{\delta_{i_1} \sqrt{n}} {\sigma\sqrt{2}},\frac{\delta_{i_2} \sqrt{n}} {\sigma\sqrt{2}} \bigg{)},
\end{align*}
with $i_1=1,i_2=3,i_3=2$;

\begin{align*}
\mu^{\Psi_{2,c}}=\bigg{(} \frac{\delta_{i_1}\sqrt{n}} {\sigma}, 
\frac{(\delta_{i_1}-\delta_{i_2}) \sqrt{n}} {\sigma},
\frac{(\delta_{i_1}-\delta_{i_3}) \sqrt{n}} {\sigma\sqrt{2}},\frac{(\delta_{i_2}-\delta_{i_3}) \sqrt{n}} {\sigma\sqrt{2}},
\frac{\delta_{i_1} \sqrt{n}} {\sigma\sqrt{2}},\frac{\delta_{i_2} \sqrt{n}} {\sigma\sqrt{2}} \bigg{)},
\end{align*}
with $i_1=2,i_2=3,i_3=1$;

\begin{align*}
\mu^{\Psi_{2,d}}=\bigg{(} \frac{\delta_{i_1}\sqrt{n}} {\sigma}, 
\frac{(\delta_{i_1}-\delta_{i_2}) \sqrt{n}} {\sigma},
\frac{(\delta_{i_1}-\delta_{i_3}) \sqrt{n}} {\sigma\sqrt{2}},\frac{(\delta_{i_2}-\delta_{i_3}) \sqrt{n}} {\sigma\sqrt{2}},
\frac{\delta_{i_1} \sqrt{n}} {\sigma\sqrt{2}},\frac{\delta_{i_2} \sqrt{n}} {\sigma\sqrt{2}} \bigg{)},
\end{align*}
with $i_1=3,i_2=2,i_3=1$;

\begin{align*}
\mu^{\Psi_{2,e}}=\bigg{(} \frac{\delta_{i_1}\sqrt{n}} {\sigma}, 
\frac{(\delta_{i_1}-\delta_{i_2}) \sqrt{n}} {\sigma},
\frac{(\delta_{i_1}-\delta_{i_3}) \sqrt{n}} {\sigma\sqrt{2}},\frac{(\delta_{i_2}-\delta_{i_3}) \sqrt{n}} {\sigma\sqrt{2}},
\frac{\delta_{i_1} \sqrt{n}} {\sigma\sqrt{2}},\frac{\delta_{i_2} \sqrt{n}} {\sigma\sqrt{2}} \bigg{)},
\end{align*}
with $i_1=3,i_2=1,i_3=2$;

\begin{align*}
\mu^{\Psi_{2,f}}=\bigg{(} \frac{\delta_{i_1}\sqrt{n}} {\sigma}, 
\frac{(\delta_{i_1}-\delta_{i_2}) \sqrt{n}} {\sigma},
\frac{(\delta_{i_1}-\delta_{i_3}) \sqrt{n}} {\sigma\sqrt{2}},\frac{(\delta_{i_2}-\delta_{i_3}) \sqrt{n}} {\sigma\sqrt{2}},
\frac{\delta_{i_1} \sqrt{n}} {\sigma\sqrt{2}},\frac{\delta_{i_2} \sqrt{n}} {\sigma\sqrt{2}} \bigg{)},
\end{align*}
with $i_1=2,i_2=1,i_3=3$ and
\begin{equation*}
\Sigma^{\Psi_2}=\begin{pmatrix}
1  & \frac{1}{2} & \frac{1}{2\sqrt{2}} & 0 &  \frac{1}{\sqrt{2}}  & \frac{1}{2\sqrt{2}}
\\
\frac{1}{2}  & 1 & \frac{1}{2\sqrt{2}} & -\frac{1}{2\sqrt{2}} & \frac{1}{2\sqrt{2}}  & -\frac{1}{2\sqrt{2}}
\\
\frac{1}{2\sqrt{2}}  & \frac{1}{2\sqrt{2}} & 1 & \frac{1}{2} & \frac{1}{2} & 0
\\
0 & -\frac{1}{2\sqrt{2}} & \frac{1}{2}  & 1 & 0 &\frac{1}{2} 
\\
\frac{1}{\sqrt{2}} & \frac{1}{2\sqrt{2}}  & \frac{1}{2} & 0 & 1 & \frac{1}{2} 
\\
\frac{1}{2\sqrt{2}} & -\frac{1}{2\sqrt{2}}  & 0 &\frac{1}{2} & \frac{1}{2} & 1
\end{pmatrix}.
\end{equation*}

$P(\Psi_3)$ equals
\begin{align*}
P(\Psi_3)& =  \int^{\infty}_{0} \int^{\infty}_{0} \int^{\infty}_{0} \int^{u_1}_{-\infty} \int^{u_1}_{-\infty} \int^{u_2}_{-\infty}  \phi\bigg{(}\mathbf{z} ,\mu^{\Psi_{3,a}},\Sigma^{\Psi_3}\bigg{)} \textbf{d}\mathbf{z}  
\\
& + \int^{\infty}_{0} \int^{\infty}_{0} \int^{\infty}_{0} \int^{\infty}_{u_1} \int^{u_1}_{-\infty} \int^{u_2}_{-\infty} \phi\bigg{(}\mathbf{z} ,\mu^{\Psi_{3,a}},\Sigma^{\Psi_3}\bigg{)} \textbf{d}\mathbf{z} 
\\
& +  \int^{\infty}_{0} \int^{\infty}_{0} \int^{\infty}_{0} \int^{u_1}_{-\infty} \int^{\infty}_{u_1} \int^{u_2}_{-\infty} \phi\bigg{(}\mathbf{z} ,\mu^{\Psi_{3,a}},\Sigma^{\Psi_3}\bigg{)} \textbf{d}\mathbf{z} 
\\
& +  \int^{\infty}_{0} \int^{\infty}_{0} \int^{\infty}_{0} \int^{u_1}_{-\infty} \int^{u_1}_{-\infty} \int^{u_2}_{-\infty}  \phi\bigg{(}\mathbf{z} ,\mu^{\Psi_{3,b}},\Sigma^{\Psi_3}\bigg{)} \textbf{d}\mathbf{z}  
\\
& + \int^{\infty}_{0} \int^{\infty}_{0} \int^{\infty}_{0} \int^{\infty}_{u_1} \int^{u_1}_{-\infty} \int^{u_2}_{-\infty} \phi\bigg{(}\mathbf{z} ,\mu^{\Psi_{3,b}},\Sigma^{\Psi_3}\bigg{)} \textbf{d}\mathbf{z} 
\\
& +  \int^{\infty}_{0} \int^{\infty}_{0} \int^{\infty}_{0} \int^{u_1}_{-\infty} \int^{\infty}_{u_1} \int^{u_2}_{-\infty} \phi\bigg{(}\mathbf{z} ,\mu^{\Psi_{3,b}},\Sigma^{\Psi_3}\bigg{)} \textbf{d}\mathbf{z} 
\\
& +  \int^{\infty}_{0} \int^{\infty}_{0} \int^{\infty}_{0} \int^{u_1}_{-\infty} \int^{u_1}_{-\infty} \int^{u_2}_{-\infty}  \phi\bigg{(}\mathbf{z} ,\mu^{\Psi_{3,c}},\Sigma^{\Psi_3}\bigg{)} \textbf{d}\mathbf{z}  
\\
& + \int^{\infty}_{0} \int^{\infty}_{0} \int^{\infty}_{0} \int^{\infty}_{u_1} \int^{u_1}_{-\infty} \int^{u_2}_{-\infty} \phi\bigg{(}\mathbf{z} ,\mu^{\Psi_{3,c}},\Sigma^{\Psi_3}\bigg{)} \textbf{d}\mathbf{z} 
\\
& +  \int^{\infty}_{0} \int^{\infty}_{0} \int^{\infty}_{0} \int^{u_1}_{-\infty} \int^{\infty}_{u_1} \int^{u_2}_{-\infty} \phi\bigg{(}\mathbf{z} ,\mu^{\Psi_{3,c}},\Sigma^{\Psi_3}\bigg{)} \textbf{d}\mathbf{z} 
\\
& +  \int^{\infty}_{0} \int^{\infty}_{0} \int^{\infty}_{0} \int^{u_1}_{-\infty} \int^{u_1}_{-\infty} \int^{u_2}_{-\infty}  \phi\bigg{(}\mathbf{z} ,\mu^{\Psi_{3,d}},\Sigma^{\Psi_3}\bigg{)} \textbf{d}\mathbf{z}  
\\
& + \int^{\infty}_{0} \int^{\infty}_{0} \int^{\infty}_{0} \int^{\infty}_{u_1} \int^{u_1}_{-\infty} \int^{u_2}_{-\infty} \phi\bigg{(}\mathbf{z} ,\mu^{\Psi_{3,d}},\Sigma^{\Psi_3}\bigg{)} \textbf{d}\mathbf{z} 
\\
& +  \int^{\infty}_{0} \int^{\infty}_{0} \int^{\infty}_{0} \int^{u_1}_{-\infty} \int^{\infty}_{u_1} \int^{u_2}_{-\infty} \phi\bigg{(}\mathbf{z} ,\mu^{\Psi_{3,d}},\Sigma^{\Psi_3}\bigg{)} \textbf{d}\mathbf{z} 
\\
& +  \int^{\infty}_{0} \int^{\infty}_{0} \int^{\infty}_{0} \int^{u_1}_{-\infty} \int^{u_1}_{-\infty} \int^{u_2}_{-\infty}  \phi\bigg{(}\mathbf{z} ,\mu^{\Psi_{3,e}},\Sigma^{\Psi_3}\bigg{)} \textbf{d}\mathbf{z}  
\\
& + \int^{\infty}_{0} \int^{\infty}_{0} \int^{\infty}_{0} \int^{\infty}_{u_1} \int^{u_1}_{-\infty} \int^{u_2}_{-\infty} \phi\bigg{(}\mathbf{z} ,\mu^{\Psi_{3,e}},\Sigma^{\Psi_3}\bigg{)} \textbf{d}\mathbf{z} 
\\
& +  \int^{\infty}_{0} \int^{\infty}_{0} \int^{\infty}_{0} \int^{u_1}_{-\infty} \int^{\infty}_{u_1} \int^{u_2}_{-\infty} \phi\bigg{(}\mathbf{z} ,\mu^{\Psi_{3,e}},\Sigma^{\Psi_3}\bigg{)} \textbf{d}\mathbf{z} 
\\
& +  \int^{\infty}_{0} \int^{\infty}_{0} \int^{\infty}_{0} \int^{u_1}_{-\infty} \int^{u_1}_{-\infty} \int^{u_2}_{-\infty}  \phi\bigg{(}\mathbf{z} ,\mu^{\Psi_{3,f}},\Sigma^{\Psi_3}\bigg{)} \textbf{d}\mathbf{z}  
\\
& + \int^{\infty}_{0} \int^{\infty}_{0} \int^{\infty}_{0} \int^{\infty}_{u_1} \int^{u_1}_{-\infty} \int^{u_2}_{-\infty} \phi\bigg{(}\mathbf{z} ,\mu^{\Psi_{3,f}},\Sigma^{\Psi_3}\bigg{)} \textbf{d}\mathbf{z} 
\\
& +  \int^{\infty}_{0} \int^{\infty}_{0} \int^{\infty}_{0} \int^{u_1}_{-\infty} \int^{\infty}_{u_1} \int^{u_2}_{-\infty} \phi\bigg{(}\mathbf{z} ,\mu^{\Psi_{3,f}},\Sigma^{\Psi_3}\bigg{)} \textbf{d}\mathbf{z},
\end{align*}

where
\begin{align*}
\mu^{\Psi_{3,a}}=\bigg{(}
\frac{(\delta_{i_1}-\delta_{i_3}) \sqrt{n}} {\sigma\sqrt{2}},\frac{(\delta_{i_2}-\delta_{i_3}) \sqrt{n}} {\sigma\sqrt{2}}, \frac{(\delta_{i_1}-\delta_{i_2}) \sqrt{n}} {\sigma},
\frac{\delta_{i_1} \sqrt{n}} {\sigma\sqrt{2}},\frac{\delta_{i_2} \sqrt{n}} {\sigma\sqrt{2}}, \frac{\delta_{i_1} \sqrt{n}} {\sigma} \bigg{)},
\end{align*}
with $i_1=1,i_2=2,i_3=3$;

\begin{align*}
\mu^{\Psi_{3,b}}=\bigg{(}
\frac{(\delta_{i_1}-\delta_{i_3}) \sqrt{n}} {\sigma\sqrt{2}},\frac{(\delta_{i_2}-\delta_{i_3}) \sqrt{n}} {\sigma\sqrt{2}}, \frac{(\delta_{i_1}-\delta_{i_2}) \sqrt{n}} {\sigma},
\frac{\delta_{i_1} \sqrt{n}} {\sigma\sqrt{2}},\frac{\delta_{i_2} \sqrt{n}} {\sigma\sqrt{2}}, \frac{\delta_{i_1} \sqrt{n}} {\sigma} \bigg{)},
\end{align*}
with $i_1=1,i_2=3,i_3=2$;

\begin{align*}
\mu^{\Psi_{3,c}}=\bigg{(}
\frac{(\delta_{i_1}-\delta_{i_3}) \sqrt{n}} {\sigma\sqrt{2}},\frac{(\delta_{i_2}-\delta_{i_3}) \sqrt{n}} {\sigma\sqrt{2}}, \frac{(\delta_{i_1}-\delta_{i_2}) \sqrt{n}} {\sigma},
\frac{\delta_{i_1} \sqrt{n}} {\sigma\sqrt{2}},\frac{\delta_{i_2} \sqrt{n}} {\sigma\sqrt{2}}, \frac{\delta_{i_1} \sqrt{n}} {\sigma} \bigg{)},
\end{align*}
with $i_1=2,i_2=3,i_3=1$;

\begin{align*}
\mu^{\Psi_{3,d}}=\bigg{(}
\frac{(\delta_{i_1}-\delta_{i_3}) \sqrt{n}} {\sigma\sqrt{2}},\frac{(\delta_{i_2}-\delta_{i_3}) \sqrt{n}} {\sigma\sqrt{2}}, \frac{(\delta_{i_1}-\delta_{i_2}) \sqrt{n}} {\sigma},
\frac{\delta_{i_1} \sqrt{n}} {\sigma\sqrt{2}},\frac{\delta_{i_2} \sqrt{n}} {\sigma\sqrt{2}}, \frac{\delta_{i_1} \sqrt{n}} {\sigma} \bigg{)},
\end{align*}
with $i_1=3,i_2=2,i_3=1$;

\begin{align*}
\mu^{\Psi_{3,e}}=\bigg{(}
\frac{(\delta_{i_1}-\delta_{i_3}) \sqrt{n}} {\sigma\sqrt{2}},\frac{(\delta_{i_2}-\delta_{i_3}) \sqrt{n}} {\sigma\sqrt{2}}, \frac{(\delta_{i_1}-\delta_{i_2}) \sqrt{n}} {\sigma},
\frac{\delta_{i_1} \sqrt{n}} {\sigma\sqrt{2}},\frac{\delta_{i_2} \sqrt{n}} {\sigma\sqrt{2}}, \frac{\delta_{i_1} \sqrt{n}} {\sigma} \bigg{)},
\end{align*}
with $i_1=3,i_2=1,i_3=2$;

\begin{align*}
\mu^{\Psi_{3,f}}=\bigg{(}
\frac{(\delta_{i_1}-\delta_{i_3}) \sqrt{n}} {\sigma\sqrt{2}},\frac{(\delta_{i_2}-\delta_{i_3}) \sqrt{n}} {\sigma\sqrt{2}}, \frac{(\delta_{i_1}-\delta_{i_2}) \sqrt{n}} {\sigma},
\frac{\delta_{i_1} \sqrt{n}} {\sigma\sqrt{2}},\frac{\delta_{i_2} \sqrt{n}} {\sigma\sqrt{2}}, \frac{\delta_{i_1} \sqrt{n}} {\sigma} \bigg{)},
\end{align*}
with $i_1=2,i_2=1,i_3=3$ and
\begin{equation*}
\Sigma^{\Psi_3}=\begin{pmatrix}
1  & \frac{1}{2} & \frac{1}{2\sqrt{2}} & \frac{1}{2} & 0  & \frac{1}{2\sqrt{2}}
\\
\frac{1}{2}  & 1 & -\frac{1}{2\sqrt{2}} & 0 & \frac{1}{2}  & 0
\\
\frac{1}{2\sqrt{2}}  & -\frac{1}{2\sqrt{2}} & 1 & \frac{1}{2\sqrt{2}} & -\frac{1}{2\sqrt{2}} &\frac{1}{2}
\\
\frac{1}{2} & 0 & \frac{1}{2\sqrt{2}}  & 1 & \frac{1}{2} &\frac{1}{\sqrt{2}} 
\\
0 & \frac{1}{2} & -\frac{1}{2\sqrt{2}} & \frac{1}{2} & 1 & \frac{1}{2\sqrt{2}} 
\\
\frac{1}{2\sqrt{2}} & 0  & \frac{1}{2} & \frac{1}{\sqrt{2}} & \frac{1}{2\sqrt{2}} & 1
\end{pmatrix}.
\end{equation*}

The expected sample size is therefore 
\begin{equation*}
E(N|\Delta)=
\sum^{3}_{j=1} \bigg{(} P(\Psi_j)  \bigg{(} \sum^{j-1}_{i=1} in + (K-j+2)jn  \bigg{)} \bigg{)}.
\end{equation*}

\color{black}
\subsection{Equations for the proposed approach to calculate the type I error under the global null}

To calculate the type I error under the global null we calculate the probability that, without loss of generality, treatment 1 is found superior to the control at a given stage $j$ and the trial stops at that given stage $j$. We define this event as $\nu_j$. This is therefore similar to $\Phi_1, \Phi_2, \Phi_3$. $P(\nu_1)$ equals
\begin{align*}
P(\nu_1)& = 2 \bigg{[} \int^{\infty}_{u_1} \int^{\infty}_{0} \int^{\infty}_{u_1} \int^{\infty}_{0} \phi\bigg{(}\mathbf{z} ,\mu^{\nu_{1,a}},\Sigma^{\nu_{1,a}}\bigg{)} \textbf{d}\mathbf{z} \bigg{]} + 
\int^{\infty}_{u_1} \int^{0}_{-\infty}\int^{0}_{-\infty}  \int^{\infty}_{u_1} \int^{\infty}_{u_1}  \phi\bigg{(}\mathbf{z} ,\mu^{\nu_{1,b}},\Sigma^{\nu_{1,b}}\bigg{)} \textbf{d}\mathbf{z} ,
\end{align*}

where
\begin{align*}
\mu^{\nu_{1,a}}=&\bigg{(} 0,0,0,0\bigg{)},
\\
\mu^{\nu_{1,b}}=&\bigg{(} 0,0,0,0,0\bigg{)},
\end{align*}
and
\begin{align*}
\Sigma^{\nu_{1,a}}=&\begin{pmatrix}
1   & \frac{1}{2} & \frac{1}{2}  & 0
\\
\frac{1}{2}  & 1 & 0 & \frac{1}{2}
\\
\frac{1}{2}  &  0 & 1 & \frac{1}{2} 
\\
0   & \frac{1}{2} & \frac{1}{2} & 1 
\end{pmatrix}.
\\
\Sigma^{\nu_{1,b}}=&\begin{pmatrix}
1   & \frac{1}{2} & \frac{1}{2}  & \frac{1}{2} &\frac{1}{2}
\\
\frac{1}{2}  & 1 & \frac{1}{2} & -\frac{1}{2} & 0
\\
\frac{1}{2}  &  \frac{1}{2} & 1 & 0& -\frac{1}{2} 
\\
\frac{1}{2}   & -\frac{1}{2} & 0 & 1 & \frac{1}{2}
\\
\frac{1}{2} & 0 & -\frac{1}{2}  & \frac{1}{2} & 1
\end{pmatrix}.
\end{align*}

$P(\nu_2)$ equals
\begin{align*}
P(\nu_2)& = 2 \bigg{[} \int^{\infty}_{u_2} \int^{\infty}_{0} \int^{\infty}_{0} \int^{u_1}_{-\infty} \int^{u_1}_{-\infty} \phi\bigg{(}\mathbf{z} ,\mu^{\nu_2},\Sigma^{\nu_2}\bigg{)} \textbf{d}\mathbf{z} 
\\
&  + \int^{\infty}_{u_2} \int^{\infty}_{0}  \int^{\infty}_{0} \int^{\infty}_{u_1} \int^{u_1}_{-\infty} \phi\bigg{(}\mathbf{z} ,\mu^{\nu_2},\Sigma^{\nu_2}\bigg{)} \textbf{d}\mathbf{z} 
\\
&  + \int^{\infty}_{u_2} \int^{\infty}_{0}  \int^{\infty}_{0} \int^{u_1}_{-\infty} \int^{\infty}_{u_1} \phi\bigg{(}\mathbf{z} ,\mu^{\nu_2},\Sigma^{\nu_2}\bigg{)} \textbf{d}\mathbf{z} \bigg{]}, 
\end{align*}

where
\begin{align*}
\mu^{\nu_2}=\bigg{(} 0,0,0,0,0 \bigg{)},
\end{align*}
and
\begin{equation*}
\Sigma^{\nu_2}=\begin{pmatrix}
1   & \frac{1}{2\sqrt{2}} & 0 &  \frac{1}{\sqrt{2}}  & \frac{1}{2\sqrt{2}}

\\
\frac{1}{2\sqrt{2}}   & 1 & \frac{1}{2} & \frac{1}{2} & 0
\\
0  & \frac{1}{2}  & 1 & 0 &\frac{1}{2} 
\\
\frac{1}{\sqrt{2}}   & \frac{1}{2} & 0 & 1 & \frac{1}{2} 
\\
\frac{1}{2\sqrt{2}}  & 0 &\frac{1}{2} & \frac{1}{2} & 1
\end{pmatrix}.
\end{equation*}
$P(\nu_3)$ equals
\begin{align*}
P(\nu_3)& = 2 \bigg{[} \int^{\infty}_{u_3} \int^{\infty}_{0} \int^{\infty}_{0} \int^{\infty}_{0} \int^{u_1}_{-\infty} \int^{u_1}_{-\infty} \int^{u_2}_{-\infty}  \phi\bigg{(}\mathbf{z} ,\mu^{\nu_3},\Sigma^{\nu_3}\bigg{)} \textbf{d}\mathbf{z}  
\\
&  + \int^{\infty}_{u_3} \int^{\infty}_{0} \int^{\infty}_{0} \int^{\infty}_{0} \int^{\infty}_{u_1} \int^{u_1}_{-\infty} \int^{u_2}_{-\infty}  \phi\bigg{(}\mathbf{z} ,\mu^{\nu_3},\Sigma^{\nu_3}\bigg{)} \textbf{d}\mathbf{z}
\\
&  + \int^{\infty}_{u_3} \int^{\infty}_{0} \int^{\infty}_{0} \int^{\infty}_{0} \int^{u_1}_{-\infty} \int^{\infty}_{u_1} \int^{u_2}_{-\infty} \phi\bigg{(}\mathbf{z} ,\mu^{\nu_3},\Sigma^{\nu_3}\bigg{)} \textbf{d}\mathbf{z} \bigg{]}, 
\end{align*}

where
\begin{align*}
\mu^{\nu_3}=\bigg{(} 0, 0,0,
0,
0,0, 0\bigg{)},
\end{align*}
and
\begin{equation*}
\Sigma^{\nu_3}=\begin{pmatrix}
1  & \frac{1}{2\sqrt{3}} & 0 & \frac{1}{\sqrt{6}} & \frac{1}{\sqrt{3}} & \frac{1}{2\sqrt{3}} &  \frac{\sqrt{2}}{\sqrt{3}}
\\
\frac{1}{2\sqrt{3}} & 1 & \frac{1}{2} & \frac{1}{2\sqrt{2}} &  \frac{1}{2} & 0  & \frac{1}{2\sqrt{2}}
\\
0 & \frac{1}{2} & 1 & -\frac{1}{2\sqrt{2}}  & 0 & \frac{1}{2} & 0
\\
\frac{1}{\sqrt{6}} & \frac{1}{2\sqrt{2}} & -\frac{1}{2\sqrt{2}} & 1 & \frac{1}{2\sqrt{2}}   & -\frac{1}{2\sqrt{2}}  &\frac{1}{2} 
\\
\frac{1}{\sqrt{3}} &\frac{1}{2} & 0 & \frac{1}{2\sqrt{2}} & 1 & \frac{1}{2} & \frac{1}{\sqrt{2}} 
\\
\frac{1}{2\sqrt{3}} &0 & \frac{1}{2}  & -\frac{1}{2\sqrt{2}} &\frac{1}{2} & 1 & \frac{1}{2\sqrt{2}}
\\
\frac{\sqrt{2}}{\sqrt{3}} & \frac{1}{2\sqrt{2}} & 0 &  \frac{1}{2} & \frac{1}{\sqrt{2}} & \frac{1}{2\sqrt{2}} & 1
\end{pmatrix}.
\end{equation*}
The type I error for a given treatment under the global null for the motivating example is
\begin{equation*}
\sum^{3}_{j=1} P(\nu_j).
\end{equation*}

\subsection{Simulations of type I error}
Figure \ref{Fig} gives the type I error for treatment 1 under multiple different values of $\theta_1$, $\theta_2$ and $\theta_3$. Values tested for $\theta_2$, $\theta_3$ are in the range of $-2\theta'$ to $2\theta'$. Values tested for $\theta_1$ are in the range of $-\theta'$ to $0$. 1,000,000 simulations of each scenario are run. The maximum value for type I error for treatment 1 is 0.02496 which is when $\theta_1=0$, $\theta_2=-2\theta'$ and $\theta_3=-2\theta'$.

\begin{figure}[h]
    \centering

    \begin{minipage}{0.48\textwidth}
        \centering
        \includegraphics[width=\textwidth]{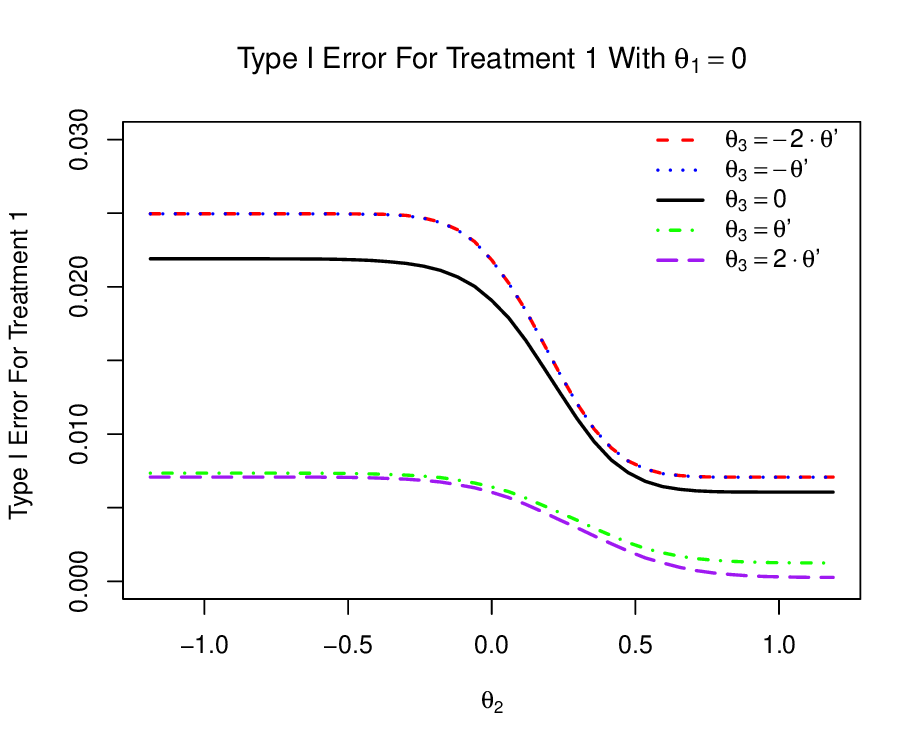}
        \caption*{(a)}
    \end{minipage}
    \hfill
    \begin{minipage}{0.48\textwidth}
        \centering
        \includegraphics[width=\textwidth]{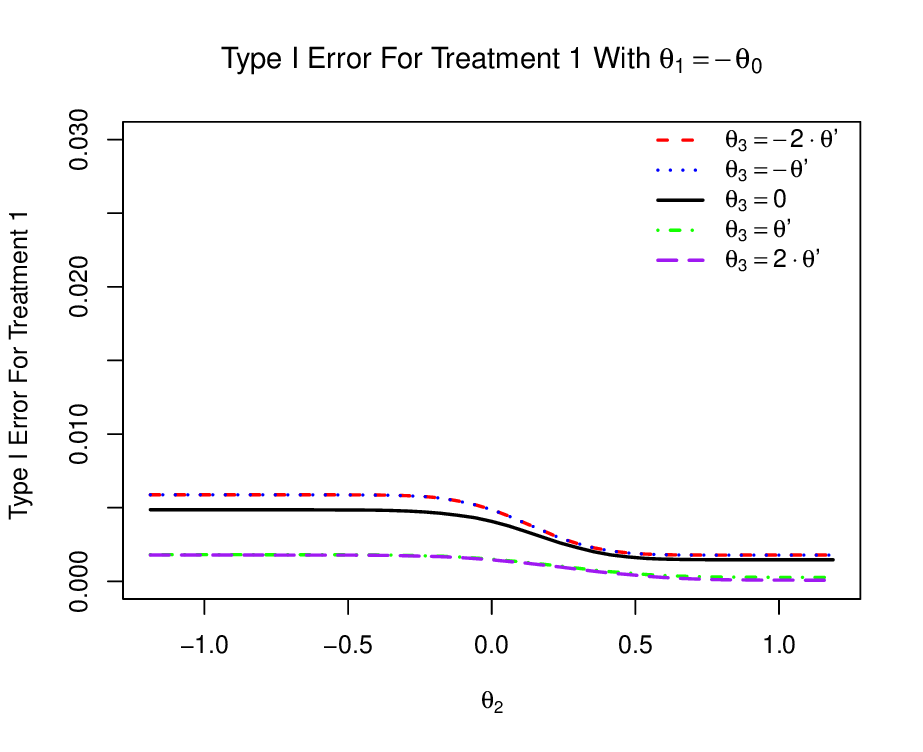}
        \caption*{(b)}
    \end{minipage}
    \hfill
    \begin{minipage}{0.48\textwidth}
        \centering
        \includegraphics[width=\textwidth]{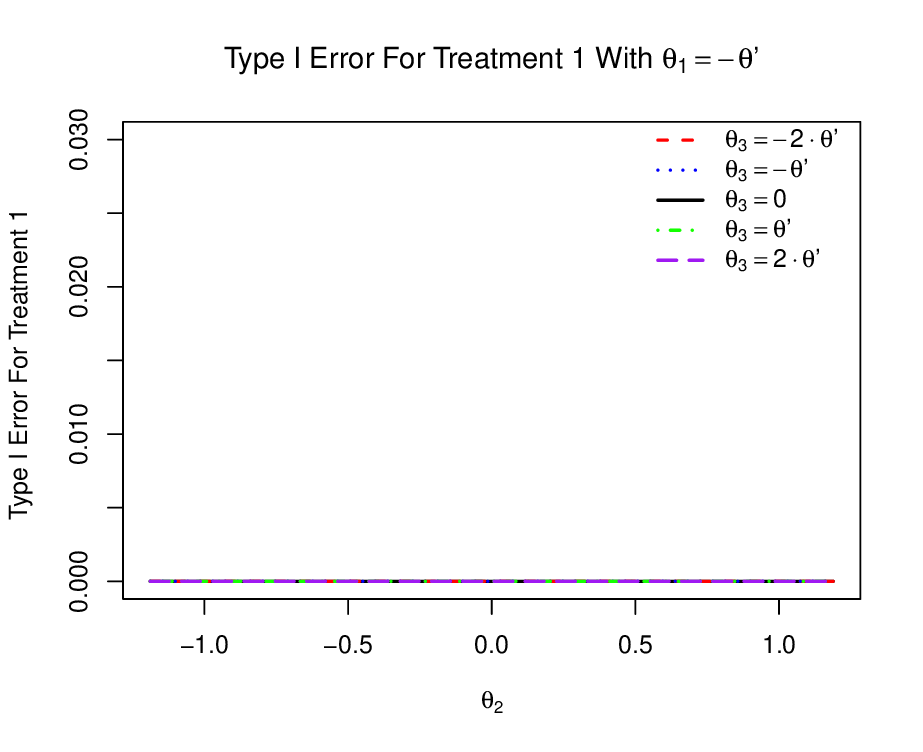}
        \caption*{(c)}
    \end{minipage}

    \caption{Type I error for treatment 1 under different values of  $\theta_1$, $\theta_2$ and $\theta_3$.}
    \label{Fig}
\end{figure}
\color{black}

\end{document}